\documentclass[acmsmall,screen,nonacm]{acmart}
\usepackage{subcaption} 

\usepackage[utf8]{inputenc}
\usepackage[T1]{fontenc}
\usepackage{microtype}
\usepackage{cleveref}
\crefformat{footnote}{#2\footnotemark[#1]#3}
\usepackage{stmaryrd}
\usepackage{amsmath}
\usepackage{graphicx}
\usepackage{mathtools}
\usepackage{cancel} 
\usepackage[dvipsnames]{xcolor}
\usepackage{wrapfig}
\usepackage{interval}
\intervalconfig{ soft open fences } 
\usepackage{tikz-cd}
\usetikzlibrary{shapes.multipart}
\usetikzlibrary{patterns}
\usetikzlibrary{patterns.meta}
\usetikzlibrary{positioning}
\usetikzlibrary{decorations.pathreplacing,calligraphy}

\usepackage{xspace}
\usepackage{multirow}

\usepackage{lib/mathpartir}
\usepackage{lib/locallabel}
\usepackage{lib/pftools}
\usepackage{lib/monotone}
\usepackage{lib/iris}

\renewcommand{\implies}{\Rightarrow}

\renewcommand{\phi}{\varphi}
\renewcommand{\theta}{\vartheta}

\renewcommand{\epsilon}{\varepsilon}

\newcommand{\isetsep}{\;\ifnum\currentgrouptype=16 \middle\fi|\;}

\usepackage{listings}
\usepackage{lib/lstasm}
\DeclareRobustCommand{\AA}{\ensuremath{\mathbb{A}}}

\DeclareRobustCommand{\NN}{\ensuremath{\mathbb{N}}}

\DeclareRobustCommand{\ZZ}{\ensuremath{\mathbb{Z}}}

\DeclareRobustCommand{\El}{\ensuremath{\mathcal{E}}}

\DeclareRobustCommand{\Il}{\ensuremath{\mathcal{I}}}

\DeclareRobustCommand{\Ml}{\ensuremath{\mathcal{M}}}

\DeclareRobustCommand{\Vl}{\ensuremath{\mathcal{V}}}

\newcommand{\V}[1]{\ensuremath{\mathit{#1}}}
\newcommand{\X}[1]{\ensuremath{\mathrm{#1}}}
\newcommand{\K}[1]{\ensuremath{\mathsf{#1}}}
\newcommand{\A}[1]{\ensuremath{\mathsf{#1}}}

\newcommand{\I}[1]{\ensuremath{\mathtt{#1}}}
\newcommand{\SC}[1]{\textsc{#1}}

\newcommand{\eqdef}{\triangleq}
\newcommand{\la}{\leftarrow}
\newcommand{\ra}{\rightarrow}
\newcommand{\Ra}{\Rightarrow}

\newcommand{\ie}{\emph{i.e.,}}
\newcommand{\eg}{\emph{e.g.,}}
\newcommand{\etc}{\emph{etc.}}

\newcommand{\cheri}{CHERI}
\newcommand{\cheritree}{CHERI-TrEE}
\newcommand{\cerise}{Cerise}
\newcommand{\cerisier}{Cerisier}
\newcommand{\rocq}{Rocq}
\newcommand{\iris}{Iris}

\definecolor{ibmcolourblind1}{RGB}{100, 143, 255} 
\definecolor{ibmcolourblind2}{RGB}{120, 94, 240} 
\definecolor{ibmcolourblind3}{RGB}{220, 38, 127} 
\definecolor{ibmcolourblind4}{RGB}{254, 97, 0} 
\definecolor{ibmcolourblind5}{RGB}{255, 176, 0} 

\definecolor{newsealcolour}{RGB}{254, 97, 0}
\definecolor{newCerisierColor}{RGB}{120, 94, 240} 
\definecolor{diffColor}{RGB}{254, 97, 0} 
\definecolor{commentGrey}{RGB}{128, 128, 128} 
\newcommand{\newCerisier}[1]{{\color{newCerisierColor} #1}}
\newcommand{\commentOut}[1]{{\color{commentGrey} #1}}
\newcommand{\diff}[1]{{\color{diffColor} #1}}
\newcommand{\emphc}[1]{{\color{diffColor}#1}}

\newcommand{\sreg}{\K{EC}}

\newcommand{\otype}{o}
\newcommand{\sealperm}{\V{sp}}
\newcommand{\regcap}{\ensuremath{(p,b,e,a)}}
\newcommand{\sealrange}{\ensuremath{\left[\sealperm,\otype_{b},\otype_{e},\otype_{a}\right]}}
\newcommand{\sealcap}{\ensuremath{ \sealed{\otype}{\V{sc}}}}
\DeclareMathOperator{\sealPred}{\X{seal\_pred}}
\DeclareMathOperator{\safeseal}{\X{safe\_to\_seal}}
\DeclareMathOperator{\safeunseal}{\X{safe\_to\_unseal}}
\DeclareMathOperator{\safeattest}{\X{safe\_to\_deinit}}
\DeclareMathOperator{\hash}{\X{hash}}

\newcommand{\hashconcat}{\mathrel{||}}

\newcommand{\enclaveLive}[2]{ \V{#1} \mapsto_{\X{E}} \V{#2} }
\DeclareMathOperator{\enclaveLiveAuth}{\X{enclavesLive}}
\newcommand{\enclaveHist}[2]{ \V{#1} \mapsto^{\always}_{\X{E}} \V{#2} }
\DeclareMathOperator{\enclaveHistAuth}{\X{enclavesHist}}
\newcommand{\enclavePrev}[1]{\X{DeInitialized}(\V{#1})}
\DeclareMathOperator{\enclavePrevAuth}{\X{enclavesPrev}}

\DeclareMathOperator{\tidxofot}{\X{tidx\_of\_otype}}
\DeclareMathOperator{\sweepr}{\X{sweep}}
\DeclareMathOperator{\freshtidx}{\X{fresh\_tidx}}
\DeclareMathOperator{\overlap}{\X{overlap}}

\newcommand{\mathtype}[1]{\K{#1}}
\newcommand{\constrtype}[1]{\X{#1}}

\newcommand{\Addr}{\mathtype{Addr}}
\newcommand{\Perm}{\mathtype{Perm}}
\newcommand{\CCap}{\mathtype{Cap}}

\newcommand{\OType}{\mathtype{OType}}
\newcommand{\SealPerm}{\mathtype{SealPerm}}
\newcommand{\SealRange}{\mathtype{SRange}}

\newcommand{\Sealable}{\mathtype{Sealable}}
\newcommand{\SealedCap}{\mathtype{SCap}}
\newcommand{\Word}{\mathtype{Word}}

\newcommand{\TIndex}{\mathtype{TIndex}}
\newcommand{\EId}{\mathtype{EIdentity}}
\newcommand{\ENum}{\mathtype{ENum}}
\newcommand{\ETable}{\mathtype{ETable}}

\newcommand{\RegName}{\mathtype{RegName}}
\newcommand{\RegFile}{\mathtype{Regs}}
\newcommand{\Mem}{\mathtype{Mem}}

\newcommand{\ExecState}{\mathtype{ExecState}}
\newcommand{\ExecConf}{\mathtype{ExecConf}}
\newcommand{\MachineState}{\mathtype{MState}}

\newcommand{\InstructionArg}{\mathtype{InstrArg}}
\newcommand{\Instruction}{\mathtype{Instr}}

\newcommand{\Running}{\constrtype{Running}}
\newcommand{\Halted}{\constrtype{Halted}}
\newcommand{\Failed}{\constrtype{Failed}}

\newcommand{\AddrMax}{\constrtype{AddrMax}}
\newcommand{\OTypeMax}{\constrtype{OTypeMax}}

\newcommand{\reservedAddresses}{\AA}
\newcommand{\customEnclaves}{\Ml}
\newcommand{\systemInvariant}{\Il}
\newcommand{\customEnclavesContract}{\V{custom\_enclaves\_contract}}
\newcommand{\customEnclavesPred}{\V{custom\_enclaves\_predicates}}

\newcommand{\safeVl}{\Vl}
\newcommand{\safeEl}{\El}
\newcommand{\safeV}[1]{\safeVl\left( #1 \right)}
\newcommand{\safeE}[1]{\safeEl\left( #1 \right)}

\newcommand{\phyLogReg}[1]{\sim_{reg}^{\V{#1}}}
\newcommand{\phyLogMem}[1]{\sim_{mem}^{\V{#1}}}
\newcommand{\stripreg}{\X{strip\_reg}}
\newcommand{\stripword}{\X{strip}}
\newcommand{\iscurrent}{\X{is\_root}}
\newcommand{\versionMono}{\X{versionMono}}
\newcommand{\reachableCurrent}{\X{reachableCurrent}}
\newcommand{\vinit}{\V{v_{init}}}
\newcommand{\aflag}{\V{a_{flag}}}

\let\oldwp\wpre
\renewcommand*{\wpre}[2]{
  \oldwp{#1}{\begin{array}{@{}c@{}}#2\end{array}}
}

\newcommand{\ppc}{\V{p_{pc}}}
\newcommand{\bpc}{\V{b_{pc}}}
\newcommand{\epc}{\V{e_{pc}}}
\newcommand{\apc}{\V{a_{pc}}}
\newcommand{\vpc}{\V{v_{pc}}}

\newcommand{\instr}[1]{{\textup{\texttt{#1}}}}

\newcommand{\reg}[1]{\K{r_{#1}}}
\newcommand{\pc}{\K{pc}}

\newcommand{\perm}[1]{\textsc{\MakeLowercase{#1}}\xspace}
\newcommand\RW{\perm{rw}}
\newcommand\RWX{\perm{rwx}}

\newcommand\RO{\perm{ro}}
\newcommand\RX{\perm{rx}}
\newcommand\enter{\perm{e}}

\newcommand{\permflowsto}{\preccurlyeq}

\newcommand{\sealed}[2]{\{ #2 \}_{#1}}

\newcommand{\crange}[2]{\interval[right closed fence=)]{#1}{#2}}

\newcommand{\confv}{\sigma}
\newcommand{\instrsem}[1]{\llbracket #1 \rrbracket}

\newcommand{\rmapsto}{\Mapsto}
\newcommand{\amapsto}[1]{\mapsto_{#1}}
\newcommand{\ecmapsto}[1]{\texttt{EC}(#1)}
\newcommand{\pure}[1]{\left\lceil #1 \right\rceil}

\usepackage{stackengine}
\newcommand{\xrsquigarrow}[2]{%
    \mathrel{\stackunder[2pt]{\stackon[2pt]{$\rightsquigarrow$}{$\scriptscriptstyle#1$}}{%
        $\scriptscriptstyle#2$}}}

\newcommand{\contspecbase}[2]{
  \left\{ #1 \right\} \xrsquigarrow{#2}{} \scaleobj{0.7}{\bullet}}

\newcommand{\stepspecbase}[2]{
        \anglebracket{ #1 }
        \rightarrow {}
        \anglebracket{ #2 }
}

\newcommand{\lcap}[5]{(#1,#2,#3,#4)_{#5}}

\makeatletter
\newcommand*{\@rowstyle}{}
\newcommand*{\rowstyle}[1]{
  \gdef\@rowstyle{#1}%
  \@rowstyle\ignorespaces%
}
\newcolumntype{=}{
  >{\gdef\@rowstyle{}}%
}
\newcolumntype{+}{
  >{\@rowstyle}%
}
\makeatother
\newcolumntype{R}{+r}
\newcolumntype{C}{+c}
\newcolumntype{L}{+l}

\definecolor{smallcolour}{RGB}{0, 150, 0} 
\newcommand{\xxsmallhshort}[1]{\ensuremath{\scalebox{0.8}{\textcolor{smallcolour}{\ensuremath{\mathsf{(#1)}}}}}}
\newcommand{\xxsmallh}[1]{\xxsmallhshort{#1}}

\newif\ifappendix \appendixtrue
\newcommand{\appendixReplace}[2]{\ifappendix #1\else Appendix #2\fi}

\usepackage{flushend}
\begin{document}

\title{Cerisier: A Program Logic for Attestation in a Capability Machine}

\author{June Rousseau}
\orcid{0009-0003-6778-6597}
\affiliation{%
  \institution{Aarhus University}
  \city{Aarhus}
  \country{Denmark}
}
\email{june.rousseau@cs.au.dk}

\author{Denis Carnier}
\orcid{0000-0003-2148-5193}
\affiliation{%
  \institution{KU Leuven}
  \city{Leuven}
  \country{Belgium}
}
\email{denis.carnier@kuleuven.be}

\author{Thomas Van Strydonck}
\orcid{0000-0002-5262-1381}
\affiliation{%
  \institution{Fortanix}
  \city{Eindhoven}
  \country{Netherlands}
}
\email{thomas.vanstrydonck@fortanix.com}

\author{Steven Keuchel}
\orcid{0000-0001-6411-438X}
\affiliation{%
  \institution{KU Leuven}
  \city{Leuven}
  \country{Belgium}
}
\email{steven.keuchel@kuleuven.be}

\author{Dominique Devriese}
\orcid{0000-0002-3862-6856}
\affiliation{%
  \institution{KU Leuven}
  \city{Leuven}
  \country{Belgium}
}
\email{dominique.devriese@kuleuven.be}

\author{Lars Birkedal}
\orcid{0000-0003-1320-0098}
\affiliation{%
  \institution{Aarhus University}
  \city{Aarhus}
  \country{Denmark}
}
\email{birkedal@cs.au.dk}

\begin{abstract}
  A key feature in trusted computing is attestation, which allows encapsulated components (enclaves) to prove their identity to (local or remote) distrusting components.
  Reasoning about software that uses the technique requires tracking how trust evolves after successful attestation.
  This process is security-critical and non-trivial, but no existing formal verification technique supports modular reasoning about attestation of enclaves and their clients, or proving end-to-end properties for systems combining trusted, untrusted and attested code.

  We contribute Cerisier, the first program logic for modular reasoning about trusted, untrusted and attested code, fully mechanized in the Iris separation logic and the Rocq Prover.
  We formalize a recent proposal, CHERI-TrEE, to extend capability machines with enclave primitives, as an extension to the Cerise capability machine and program logic.
  Our program logic comes with a universal contract for untrusted code, which captures both capability safety and local enclave attestation.
  Like Cerise, this universal contract is phrased in terms of a logical relation defining capabilities' authority.
  We demonstrate Cerisier by proving end-to-end properties for three representative applications of trusted computing: secure outsourced computation, mutual attestation and a modeled trusted sensor component.
\end{abstract}
\maketitle

\section{Introduction}
\label{sec:introduction}

Trusted Computing refers to techniques for allowing software components to
establish trust in one another \citep{maene_hardware-based_2018}.
The technology has applications in digital
rights management, anti-cheat measures in gaming, private contact discovery, etc.
A simple representative application is client code outsourcing computation to untrusted code, for example, from a kernel client to a local user process to reduce available privilege, or from a resource-constrained device to a large but untrusted cloud server.

In this paper, we are interested in systems featuring \emph{attestation} and a \emph{dynamic root of trust}, where a
component $C$ (an \emph{enclave}) can ask the system for an attestation
certificate. This certificate attests to $C$'s identity, which might consist of a
hash of $C$'s source code, information about its location in memory and the system
it is running on. The certificate can be used to prove $C$'s identity to other components,
assuming that they trust the system's Trusted Computing Base (TCB), which typically consists only of
hardware and firmware, but not the operating system
or application software. The certificate can be securely communicated to other components on the same system (local attestation) or cryptographically signed using a system-embedded cryptographic key for transmission to remote components (remote attestation).

Attestation involves a complex interplay and non-trivial trust relationship
between trusted components, enclave components
(initially untrusted, but become trusted after successful attestation), and
surrounding untrusted code and trusted hardware/firmware. It is
typically applied in security-critical scenarios and trusted components
are kept small for security reasons. These characteristics (\emph{small} but
\emph{complex} and \emph{critical}) suggest the field as a prime target for
rigorous software verification.
Until now, verification has remained restricted to only one of
the involved components (for example, the crypto used by the trusted
hardware/firmware \cite{zain_formally-verified_2025,sardar_demystifying_2021}),
which does not yield end-to-end properties of the combined system.

In this paper, we present Cerisier: the first rigorous program logic for
modular reasoning about trusted, untrusted and attested code, and producing end-to-end guarantees about the combined system.
We base ourselves on the CHERI-TrEE~\citep{cheriTree} trusted computing primitives for the CHERI-RISC-V capability machine~\citep[\S4]{cheri-isa-v9}, and on Cerise~\citep{cerise}, a
program logic implemented in Iris~\citep{iris} for a capability
machine modeled loosely after CHERI-RISC-V.

Capability machine instruction set architectures (ISAs) offer a primitive
notion of capabilities that
represent the authority to perform certain actions, e.g., read-write
access to a certain memory range, or the permission to
invoke a piece of code with access to private data and
capabilities, without exposing that private data to the caller.

CHERI-TrEE aims to maximally reuse existing CHERI-RISC-V primitives
and only adds a minimal set of new ones to support trusted computing.
Specifically, it reuses so-called sealed capability pairs for encapsulating
components, and sealed capabilities for secure communication.
It adds an \emph{exclusive memory ownership primitive} for components to verify their exclusive access to a memory region.
The CPU guarantees this exclusive access based on exhaustively scanning the memory for overlapping capabilities.
This happens at enclave initialization time, or at the request of a component to dynamically extend its memory footprint.
This essentially enables what is known as a dynamic root-of-trust: a component can start out as untrusted code, but at some point ask the CPU to initialize itself as an enclave.
CHERI-TrEE also adds a \emph{primitive for local attestation}, through which a
component can ask the CPU to vouch for its identity. Concretely, the CPU will
provide the component with certain private capabilities that can be verifiably
connected to the component's identity by distrusting local software.

Cerisier extends the Cerise program logic with support for CHERI-TrEE: adaptations of existing techniques for supporting sealed capabilities (based on program logics for symbolic cryptography \cite{sumiiLogRelEncryption}) and the exclusive memory ownership primitive (based on a program logic for garbage collection \citep{hurSeparationLogicPresence2011}).\footnote{We do not support sealed pairs for encapsulation, since Cerise already offers an alternative: sentry capabilities.}
Additionally, Cerisier adds a novel technique for reasoning about CHERI-TrEE's attestation feature by extending Cerise's universal contract.
This universal contract provides guarantees about untrusted code
\citep{georges_efficient_2021,huyghebaert_formalizing_2023}, formalized in terms of a logical relation.
Specifically, before using the universal contract, one can register predicates for a number of enclaves of interest, hereafter \emph{enclave sealing predicates}.
After verifying that the registered enclaves do indeed respect their registered enclave sealing predicates, the universal contract then ensures that capabilities produced by untrusted code will respect attested enclaves' enclave sealing predicates.
The resulting reasoning pattern is unusual: for completely arbitrary code, we get quite specific guarantees about the capabilities they can produce, namely that those capabilities can only be correctly attested if they satisfy the registered enclave sealing predicate (we present a concrete simple example in the next section).
The contributed technique quite generally captures intuitive reasoning on trusted execution systems and we believe it applies to general trusted computing systems even though our current development is specific to capability machines.

To summarize, in this paper we contribute:
\begin{itemize}
  \item An operational semantics for a capability machine with attestation primitives (\Cref{sec:opsem}).
  \item The first program logic with support for modular reasoning about attested enclaves, their clients and untrusted code, including support for sealing, memory revocation and local attestation (\Cref{sec:proglog}).
  \item A universal contract capturing both capability safety and local enclave attestation, formalized using a novel parameterized logical relation (\Cref{sec:enclave_res,sec:ftlr}).
  \item Three case studies inspired by the literature, in which we prove non-trivial end-to-end properties of representative applications of trusted computing (\Cref{sec:case_studies}).
\end{itemize}
All definitions and proofs are mechanized in \rocq{}~\citep{cerisier-artifact,cerisier-repository}, and available in the accompanying material~\citep{cerisier-arxiv}.
In the rest of the paper, the definitions marked in \textcolor{newCerisierColor}{blue} correspond to
novel contributions of this work, when it extends preexisting definitions.

\section{Trusted Computing Motivating Example: Secure Outsourced Computation}
\label{sec:soc}

\label{sec:motiv-exampl-secure}

To motivate our technical development, let us take a closer look at a representative application of trusted execution: a simple form of \emph{secure outsourced computation} (SOC) \citep{barbosa_foundations_2016,li_survey_2023}.
We assume that a client wants a computation-intensive task to be executed by untrusted code.
For ease of explanation, this section considers local attestation of a stateless computation task with public inputs and we hide some technical details behind pseudocode until subsequent sections.

\begin{figure}

\newcommand{\squarew}{4}
\newcommand{\squareh}{3}
\newcommand{\smallsquarew}{1.25}
\newcommand{\smallsquareh}{1}

\newcommand{\ssmallsquareh}{1}
\newcommand{\ssmallsquarew}{1.3}

\newcommand{\paddingwA}{0.6*\squarew}
\newcommand{\paddinghA}{0.3}
\definecolor{red1}{HTML}{EF2929}
\newcommand{\scalefig}{0.75}

\centering

\captionsetup[subfigure]{justification=centering}
\begin{subfigure}[t]{0.3\textwidth}
\centering
\begin{tikzpicture}[scale=\scalefig, every node/.style={transform shape}]
  \draw[pattern={Lines[angle=45,distance={6pt/sqrt(2)}]}, pattern color=red1, opacity=0.1]
  (0,0) -- (0,\squareh) -- (\squarew,\squareh) -- (\squarew,0) -- (0,0) ;
  \draw[]
  (0,0) -- (0,\squareh) -- (\squarew,\squareh) -- (\squarew,0) -- (0,0) ;

  \node at (0.35*\squareh, 0.75*\squareh) {\textbf{Adversary}};

  \draw[fill=white]
        (0,0) -- (0, \smallsquareh) -- (\smallsquarew, \smallsquareh) -- (\smallsquarew, 0) -- (0,0);
        \node at (0.5*\smallsquarew, 0.5*\smallsquareh) {\textbf{Client}};

  \draw[semithick,->,color=red]
        (0.8*\smallsquareh, 0.3*\smallsquareh) to[bend right] +(0.3*\squarew, 0);
\end{tikzpicture}
    \caption{Passing control to the adversary.}
    \label{fig:interaction_scenario_attest_tc:calling_adv}
\end{subfigure}
\begin{subfigure}[t]{0.3\textwidth}
  \centering
\begin{tikzpicture}[scale=\scalefig, every node/.style={transform shape}]
  \draw[pattern={Lines[angle=45,distance={6pt/sqrt(2)}]}, pattern color=red1, opacity=0.1]
  (0,0) -- (0,\squareh) -- (\squarew,\squareh) -- (\squarew,0) -- (0,0) ;
  \draw[]
  (0,0) -- (0,\squareh) -- (\squarew,\squareh) -- (\squarew,0) -- (0,0) ;
  \node at (0.35*\squareh, 0.75*\squareh) {\textbf{Adversary}};

  \draw[fill=white]
        (0,0) -- (0, \smallsquareh) -- (\smallsquarew, \smallsquareh) -- (\smallsquarew, 0) -- (0,0);
        \node at (0.5*\smallsquarew, 0.5*\smallsquareh) {\textbf{Client}};

  \draw[dashed, fill=white]
  (\paddingwA, 0.5*\squareh - 0.5\ssmallsquareh)
  -- (\paddingwA + \ssmallsquarew, 0.5*\squareh - 0.5\ssmallsquareh)
  -- (\paddingwA + \ssmallsquarew, 0.5*\squareh - 0.5\ssmallsquareh + \ssmallsquareh)
  -- (\paddingwA, 0.5*\squareh - 0.5\ssmallsquareh + \ssmallsquareh)
  -- (\paddingwA, 0.5*\squareh - 0.5\ssmallsquareh) ;
  \node[align=center] at ( \paddingwA + 0.5 *\ssmallsquarew, 0.5*\squareh)
  {\small \textbf{SOC}\\\textbf{enclave}};
  \draw[semithick,->,color=red]
  ( \paddingwA - 0.25*\ssmallsquarew, 0.5*\squareh + 0.35 *\ssmallsquareh)
  to[bend left] +(0.5*\ssmallsquarew, 0);
  \draw[semithick,<-,color=red]
  ( \paddingwA - 0.25*\ssmallsquarew, 0.5*\squareh - 0.35 *\ssmallsquareh)
  to[bend right] +(0.5*\ssmallsquarew, 0);

\end{tikzpicture}
    \caption{The adversary initialises the enclave.}
    \label{fig:interaction_scenario_attest_tc:adv_einit}
\end{subfigure}
  \begin{subfigure}[t]{0.3\textwidth}
  \centering
\begin{tikzpicture}[scale=\scalefig, every node/.style={transform shape}]
  \draw[pattern={Lines[angle=45,distance={6pt/sqrt(2)}]}, pattern color=red1, opacity=0.1]
  (0,0) -- (0,\squareh) -- (\squarew,\squareh) -- (\squarew,0) -- (0,0) ;
  \draw[]
  (0,0) -- (0,\squareh) -- (\squarew,\squareh) -- (\squarew,0) -- (0,0) ;
  \node at (0.35*\squareh, 0.75*\squareh) {\textbf{Adversary}};

  \draw[fill=white]
        (0,0) -- (0, \smallsquareh) -- (\smallsquarew, \smallsquareh) -- (\smallsquarew, 0) -- (0,0);
        \node at (0.5*\smallsquarew, 0.5*\smallsquareh) {\textbf{Client}};

  \draw[dashed, fill=white]
  (\paddingwA, 0.5*\squareh - 0.5\ssmallsquareh)
  -- (\paddingwA + \ssmallsquarew, 0.5*\squareh - 0.5\ssmallsquareh)
  -- (\paddingwA + \ssmallsquarew, 0.5*\squareh - 0.5\ssmallsquareh + \ssmallsquareh)
  -- (\paddingwA, 0.5*\squareh - 0.5\ssmallsquareh + \ssmallsquareh)
  -- (\paddingwA, 0.5*\squareh - 0.5\ssmallsquareh) ;
  \node[align=center] at ( \paddingwA + 0.5 *\ssmallsquarew, 0.5*\squareh)
  {\small \textbf{SOC}\\\textbf{enclave}};

  \draw[thick,->,color=ForestGreen]
  (\paddingwA + 0.15*\ssmallsquarew, 0.5*\squareh - 0.35*\ssmallsquareh)
  to [bend left]
  node [midway,anchor=north,fill=white,xshift=20,yshift=10] {\small $P_{\A{SOC}}(w): w = 42$}
  (0.8*\smallsquarew, 0.2*\smallsquareh);

\end{tikzpicture}
    \caption{The client attests the enclave.
    }
    \label{fig:interaction_scenario_attest_tc:custom_attest}
  \end{subfigure}

  \caption{
  \label{fig:interactions_scenarios_attest_tc}
  Interaction scenario for the secure outsourced computation (SOC) example.}%
\end{figure}


\begin{figure}
  {\small
  \begin{subfigure}[t]{0.40\textwidth}
    \begin{equation*}
      \begin{aligned}
        \mathrm{Enclave} \triangleq\enspace
        & (k_p, k_s) \gets \mathrm{GetAttestationKeys}() \\
        & r \gets 42 \quad \commentOut{ // \mathit{\ computation...}} \\
        & \textbf{return } (\mathrm{Attest}(r, k_s),~ k_p) \\[5pt]
        \mathrm{Adversary} \triangleq\enspace
        & \mathrm{enclave}^* \gets \mathrm{Initialize}(\mathrm{Enclave}) \\
        & \textbf{return } \mathrm{enclave}()
      \end{aligned}
    \end{equation*}
  \end{subfigure}
  \hfill
  \begin{subfigure}[t]{0.55\textwidth}
    \vspace{4pt}
    \begin{equation*}
      \begin{aligned}
        \mathrm{Client} \triangleq\enspace
        & (r^*, k_p) \gets \mathrm{Adversary}() \\
        & \textbf{if } \neg\, \mathrm{VerifyKey}(k_p, \mathrm{ID}_{\mathrm{soc}})\ \textbf{ then abort} \\
        & \textbf{if } \neg\, \mathrm{VerifyMessage}(r^*, k_p)\ \textbf{ then abort} \\
        & \textbf{assert } \mathrm{Value}(r^*) = 42
      \end{aligned}
    \end{equation*}
  \end{subfigure}
  }
  \caption{\label{fig:secure_outsourced_computation_pseudocode}Simplified pseudocode of Secure
    Outsourced Computation (SOC) using CHERI-TrEE primitives.
    The $\mathrm{Enclave}$ and $\mathrm{Client}$ components contain known code. Under normal circumstances the $\mathrm{Adversary}$ executes arbitrary code, here we fix it to capture its intended execution.}
\end{figure}

In our example, a trusted client interacts with an untrusted adversary and an enclave (initially untrusted), in a scenario depicted in \Cref{fig:interactions_scenarios_attest_tc}
and implemented by the pseudocode in \Cref{fig:secure_outsourced_computation_pseudocode}.
After the trusted client invokes the adversary
(\Cref{fig:interaction_scenario_attest_tc:calling_adv}),
the adversary initializes the enclave using a trusted execution primitive (\Cref{fig:interaction_scenario_attest_tc:adv_einit}).
This initialization triggers the encapsulation of the enclave: it gets exclusive access to its own memory and can only be invoked at its entry point.
Additionally, the enclave obtains a public/private key pair $(k_p, k_s)$ for attesting messages (for now, imagine a symbolic model of cryptographic signatures).
The untrusted adversary then passes control to the enclave
(\Cref{fig:interaction_scenario_attest_tc:adv_einit}) which will perform the heavy computation and emit a message containing the attested result~$r^*$ and the public attestation key~$k_p$.

After receiving this message from the untrusted adversary acting as an intermediary, the client validates the received message (\Cref{fig:interaction_scenario_attest_tc:custom_attest}).
Successful validation allows the client to trust that the message originated from a correctly initialized enclave with a certain identity $\mathrm{ID}_{\mathrm{SOC}}$.
This identity is the combination of the enclave's executable code with details about the system it runs on, summarized using a cryptographic hash.
The client can then verify that the corresponding enclave will only attest messages containing correct results, and conclude to trust the received result.

Let us now give a very rough sketch to explain how our approach enables us
to obtain an end-to-end guarantee
for this example.

First, we must prove that the SOC enclave code respects its enclave sealing predicate.
In this case, we've picked the predicate that states the enclave will compute 42.
Since the enclave's identity---and hence its code---is known after validation by the client,
we can employ standard Hoare-style reasoning using the rules of the logic to complete the proof.
Second, we must prove the specification (Hoare triple) for the client program. Again, since the
client program code is known, we can proceed using the proof rules of the logic.
However, note that the client invokes (and later awaits) the untrusted adversary.
In general, we will not know which code the adversary executes\footnote{In \Cref{fig:secure_outsourced_computation_pseudocode} we have shown concrete example code for the adversary to make the example easier to understand.},
thus, when we arrive at the invocation of the adversary,
we make use of a \emph{universal contract}, defined using the program logic,
 to obtain a Hoare triple for the unknown adversary's code.
To clarify, the universal contract can be viewed as a specification for arbitrary code.
In our approach, we parameterize this contract by a list of
identities of enclaves and their associated predicates, which the adversary
is allowed to initialize.
In this particular example, this list only contains the identity of the SOC enclave.
We then instantiate the contract with
the proof of the SOC enclave to obtain a property for the adversary's code.
Finally, we can prove that the assert of the client holds
because the value's attest must respect the enclave sealing predicate of the SOC enclave.

This section ignores many details and assumptions for ease of explanation.
In particular, \Cref{fig:secure_outsourced_computation_pseudocode} uses high-level primitives,
but in the rest of the paper, we use the instructions of the Cerisier machine.
In the next sections, we explain some background knowledge before providing more detailed explanations
of the Cerisier machine.

\section{Background and Technical Foundations}
\label{sec:background-cm-cerise-cheritree}

In this section, we recall the Cerise capability machine model and its program logic.
Moreover, we discuss the addition of sealed capabilities, the memory revocation primitives, and
hashing primitives, which are all used as a foundation for attestation.

\subsection{Background: Capability Machines and Cerise Operational Semantics}

Hardware capabilities are one technique to isolate multiple software components executing on a single system.
The ideas date back to seminal work by~\citet{dennis1966} and~\citet{carter_hardware_1994}, yet more recently, they have garnered renewed interest through the Capability Hardware Enhanced RISC Instructions (\cheri) project~\citep{cheri}.
\cheri{} ISAs offer fine-grained memory protection and fast protection-domain crossing through \textit{capabilities}, unforgeable tokens that represent authority over memory or other resources.
Since capabilities are unforgeable, they support the secure interaction between trusted code and untrusted, potentially adversarial code, by restricting the set of capabilities the untrusted code has access to, as depicted in \Cref{fig:cerise_interaction_scenario}.

\begin{figure}
  \small
{
\begin{minipage}[t]{0.45\textwidth}
\[
\arraycolsep=3pt
\begin{array}{lclcl}
  a & \in & \Addr & \eqdef & [0, \AddrMax] \\
  p & \in & \Perm & \mathrel{::=} & \perm{o} \mid \enter \mid \RO \mid \RX \mid \RW \mid \RWX \\
  c & \in & \CCap & \eqdef & \left\{ (p, b, e, a) \mid b, e, a \in \Addr \right\} \\
  w & \in & \Word & \eqdef & \ZZ + \CCap\\
  r &\in &\RegName & \mathrel{::=} & \pc \mid \reg{0} \mid \reg{1} \mid \ldots \mid \reg{31}
\end{array}
\]
\end{minipage}%
\vspace{-0.2em}
\begin{minipage}[t]{0.48\textwidth}
\[
\arraycolsep=3pt
\begin{array}{LCLCL}
  \V{reg} & \in & \RegFile & \eqdef & \RegName \rightarrow \Word \\
  m & \in & \Mem & \eqdef & \Addr \rightarrow \Word \\
  s & \in & \ExecState & \mathrel{::=} & \Running \mid \Halted \mid \Failed \\
  \confv & \in & \ExecConf & \eqdef & \RegFile \times \Mem
  \\
  \Sigma & \in & \MachineState & \eqdef & \ExecState \times \ExecConf
\end{array}
\]
\end{minipage}%
\vspace{0.2em}
\begin{minipage}[t]{\textwidth}
\[
\arraycolsep=3pt
\begin{array}{@{}lcl@{}}
  \rho \in \InstructionArg & \eqdef & \ZZ + \RegName
\\
  i \in \Instruction & \mathrel{::=} &
  \begin{array}[t]{@{}+l@{}}
               \instr{jmp}\; r_\mathrm{d} \mid
               \instr{jnz}\; r_\mathrm{d} \; r_\mathrm{cond} \mid
               \instr{fail} \mid
               \instr{halt} \mid
               \instr{mov}\; r_\mathrm{d}\; \rho_\mathrm{s} \mid
               \instr{add}\; r_\mathrm{d} \; \rho_\mathrm{1} \; \rho_\mathrm{2} \mid
  \\
               \instr{sub}\; r_\mathrm{d} \; \rho_\mathrm{1} \; \rho_\mathrm{2} \mid
               \instr{lt}\; r_\mathrm{d} \; \rho_\mathrm{1} \; \rho_\mathrm{2} \mid
               \instr{lea}\; r_\mathrm{d} \; \rho_\mathrm{s} \mid
               \instr{load}\; r_\mathrm{d}\; r_\mathrm{s} \mid
               \instr{store}\; r_\mathrm{d} \; \rho_\mathrm{s} \mid
               \instr{restrict}\; r_\mathrm{d} \; \rho_\mathrm{s} \mid
  \\
               \instr{subseg}\; r_\mathrm{d} \; \rho_{1}\; \rho_{2} \mid
               \instr{getp}\; r_\mathrm{d}\; r_\mathrm{s} \mid
               \instr{getb}\; r_\mathrm{d}\; r_\mathrm{s} \mid
               \instr{gete}\; r_\mathrm{d}\; r_\mathrm{s} \mid
               \instr{geta}\; r_\mathrm{d}\; r_\mathrm{s}
\end{array}
\end{array}
\]
\end{minipage}
}

\caption{
  \label{fig:opsem_grammar}
  The base Cerise capability machine's words, state, and instructions.
}

\end{figure}


\Cref{fig:opsem_grammar} defines the machine state and instruction syntax of the \cerise{} capability machine, loosely modeled after \cheri{}.
Chiefly, \cerise{} carefully distinguishes two kinds of machine words: regular integers $\ZZ$ for arithmetic and capabilities $\CCap$ for memory accesses and function closures.
A \cerise{} capability is a tuple $\regcap$, representing a memory region in the interval
$\crange{b}{e}$, a permission $p$ (such as read \perm{R}, write \perm{W}, execute \perm{X}), and a memory address $a$ which the capability points to.
The machine's state is represented as a pair ($\RegFile \times \Mem$).
It has three execution states: \Running, \Halted{} or \Failed.
The register file $\V{reg} : \RegFile$ and the memory $\V{mem} : Mem$ map register names and memory addresses to machine words, respectively.
\WFclear
\begin{wrapfigure}{R}{.25\textwidth}
\newcommand{\squarew}{2}
\newcommand{\squareh}{2}
\newcommand{\smallsquarew}{1}
\newcommand{\smallsquareh}{1}
\newcommand{\paddingg}{0.09}
\definecolor{red1}{HTML}{EF2929}
\newcommand{\scalefig}{0.75}
\centering

{\begin{tikzpicture}[scale=\scalefig, every node/.style={transform shape}]
  \draw[pattern={Lines[angle=45,distance={6pt/sqrt(2)}]}, pattern color=red1, opacity=0.1]
  (0,0) -- (0,\squareh) -- (\squarew,\squareh) -- (\squarew,0) -- (0,0) ;
  \draw[]
  (0,0) -- (0,\squareh) -- (\squarew,\squareh) -- (\squarew,0) -- (0,0) ;

  \node at (0.25*\squareh, 0.75*\squareh) {\textbf{U}};

  \draw[fill=white]
        (0,0) -- (0, \smallsquareh) -- (\smallsquarew, \smallsquareh) -- (\smallsquarew, 0) -- (0,0);

  \node at (0.5*\smallsquareh, 0.5*\smallsquareh) {\textbf{K}};

  \draw[semithick,->,color=red]
        (0.8*\smallsquareh, 0.3*\smallsquareh) to[bend right] +(0.3*\squarew, 0);
\end{tikzpicture}
\subcaption{ \label{fig:interaction_scenario:calling_adv}
  Scenario 1: passing control to untrusted code}
}

{
\begin{tikzpicture}[scale=\scalefig, every node/.style={transform shape}]
  \draw[pattern={Lines[angle=45,distance={6pt/sqrt(2)}]}, pattern color=red1, opacity=0.1]
  (0,0) -- (0,\squareh) -- (\squarew,\squareh) -- (\squarew,0) -- (0,0) ;
  \draw[]
  (0,0) -- (0,\squareh) -- (\squarew,\squareh) -- (\squarew,0) -- (0,0) ;

  \node at (0.25*\squareh, 0.75*\squareh) {\textbf{U}};

  \draw[fill=white]
        (0,0) -- (0, \smallsquareh) -- (\smallsquarew, \smallsquareh) -- (\smallsquarew, 0) -- (0,0);

  \node at (0.5*\smallsquareh, 0.5*\smallsquareh) {\textbf{K}};

  \draw[semithick,->,color=red,dashed]
        (0.8*\smallsquareh, 0.4*\smallsquareh) to[bend right] +(0.3*\squarew, 0);

  \draw[semithick,->,color=red]
        (0.8*\smallsquareh+0.30*\squarew, 0.7*\smallsquareh) to[bend right] +(-0.3*\squarew, 0);
\end{tikzpicture}
\subcaption{\label{fig:interaction_scenario:being_called_by_adv}
  Scenario 2: being called by untrusted code (possibly many times)}
}

\caption{\label{fig:cerise_interaction_scenario}Two scenarios: a trusted component (K) interacts with an untrusted context (U). }

\end{wrapfigure}

\begin{figure}

  \footnotesize
  \begin{mathpar}
  \inferrule[ExecSingle\label{ExecSingle}]{}
  {{\begin{array}{l}
    (\K{Running}, \confv) \rightarrow \left\{
    \begin{array}{l}
      \instrsem{\mathit{decode}(z)}(\confv)
      \quad \begin{array}[t]{ll}
        \mathrm{if}  \!\!\!\!&
        \confv.\mathrm{reg(\pc)} = (p, b, e, a) \wedge b \le a < e \wedge  p \in \{\RX, \RWX \} \wedge \confv.\mathrm{mem(\V{a})} = z \\[0.5em]
      \end{array}
      \\
      (\K{Failed},\, \confv) \qquad\quad\;\, \mathrm{otherwise}
    \end{array}
    \right.
    \end{array}}
  }
  \end{mathpar}
  \vspace{1em}

\small
\bgroup
\setlength\tabcolsep{0.54em}
\footnotesize
\noindent\makebox[\textwidth]{\begin{tabular}{|C|L|L|}
  \hline
  $i$ & $\instrsem{i}(\confv)$ & Conditions \\ \hline
  \ldots & \ldots & \ldots
  \\ \hline
  $\instr{load}\; r_d\; r_s$ & $\X{updPC}(\confv[\X{reg}.r_d \mapsto w])$
  & \!\!\!\begin{tabular}{l}
      $\confv.\X{reg}(r_s) = (p, b, e, a)$ and $w = \confv.\X{mem}(a)$ \\
      and $b \le a < e$ and
      $p \in \{ \RO,  \RX, \RW, \RWX \}$
    \end{tabular}
  \\ \hline
  $\instr{jmp}\; r$
  & {$(\K{Running}, \confv \left[ \X{reg}.\pc \mapsto \V{new} \right])$}
                               & $ \text{if } \confv.\X{reg}(r) = (\enter, b, e, a)
                                 \text{ then } \V{new} =(\perm{RX}, b, e, a)
                                 \text{ else } \V{new} =\confv.\X{reg}(r)$
  \\ \hline
  $\instr{getb}\; r_d\; r_s$
  & $\X{updPC}(\confv[\X{reg}.r_d \mapsto b])$
  & \!\!\!\begin{tabular}{l}
      $\confv.\X{reg}(r_s) = (\_, b, \_, \_)$
    \end{tabular}
  \\ \hline
\end{tabular}}
\egroup

\footnotesize
\begin{mathpar}
  \X{updPC}(\confv) = \left\{
    \begin{array}{ll}
      (\K{Running}, \confv[\X{reg}.\pc \mapsto (p, b, e, a+1)])
      & \text{if } \confv.\X{reg}(\pc) = (p, b, e, a)\\ 
      (\K{Failed}, \confv) & \text{otherwise}
    \end{array}
    \right.

    \mathit{decode} : \ZZ \rightarrow \Instruction
\color{black}
\end{mathpar}

\caption{\label{fig:opsem_excerpt}
Excerpt of the operational semantics for of a single instruction.
The complete operational semantics can be found in the Supplemental Materials, under
\appendixReplace{\Cref{sec:appendix_opsem}}{A}.
}
\end{figure}


\Cref{fig:opsem_excerpt} shows an excerpt of the operational semantics.
At each step of execution, the rule \textsc{ExecSingle} is applied if the machine is in the \Running{} state.
The machine first checks if the PC contains a valid, executable capability.
It then decodes the instruction $i$ stored at current address $a$, and updates machine state according to the instruction's semantics $\instrsem{i}(\sigma)$, detailed in the table.
Like other capability machines, \cerise{} instructions ensure that capabilities are unforgeable, i.e. they can only be derived from other capabilities with more authority.
Memory accesses can only be performed with memory capabilities, and the machine dynamically checks the access against their range and permissions.
For instance, executing the $\instr{load}\;\reg{d}\;\reg{s}$ instruction requires a capability in source register $\reg{s}$ with at least $\perm{R}$ permission, and an in-bounds address $a \in \crange{b}{e}$.
If so, the register file is updated to map destination register $\reg{d}$ to the word in memory at address $a$.

Controlled invocation is achieved with so-called \emph{sealed entry (sentry)} capabilities, with permission $\perm{E}$.
They can be thought of as opaque and immutable, only to be jumped to, roughly a low-level encapsulated closure.
For instance, the $\instr{jmp}$ instruction will ``upgrade'' a capability's $\perm{E}$ permission to  $\perm{RX}$ when invoked, thus giving the invokee access to the capability's memory range (where private state may be stored).

\subsection{Cerise Program Logic}

For reasoning, we start from the \cerise{} program logic \citep{cerise},
formalized and mechanized in \iris{}~\citep{iris} and \rocq{}.
It supports modular reasoning about systems combining both trusted and adversarial code.

The syntax of the program logic is presented in \Cref{fig:program_logic_syntax} and extends
\iris{}, an impredicative higher-order separation logic.
\begin{figure} 
\[
  \begin{array}{l}
    P, Q \in \iProp \mathrel{::=}
    \begin{array}[t]{lr}
      \TRUE \mid \FALSE \mid \forall x \ldotp P \mid \exists x \ldotp P \mid \ldots
      & \text{higher-order logic}
      \\
      \mid P \ast Q \mid P \wand Q \mid \pure{\phi} \mid \always P \mid \later P
      \mid \knowInv{}{P}
      & \text{separation logic and invariants}
      \\
      \mid \A{a} \mapsto \V{lw}
      \mid r \rmapsto \V{lw}
      & \text{machine resources}
      \\
      \mid \stepspecbase{P}{s\ldotp Q}{}
      \mid \contspecbase{P}{}
      & \text{program logic}
    \end{array}
  \end{array}
\]
\caption{Syntax of the \cerisier{} program logic.}
\label{fig:program_logic_syntax}
\end{figure}

%
%
It contains the usual higher-order logic quantifiers and connectives:
universal $\forall$ and existential $\exists$ quantification,
  conjunction $\land$ and disjunction $\lor$,~\etc{}, as well as
the separating conjunction $\ast$ and magic wand $\wand$.
The pure proposition $\pure{\varphi}$ holds
if the meta-proposition $\varphi$ holds in the meta-logic
(in our case, \rocq).
\iris{} comes with two kinds of propositions: ephemeral and persistent
propositions. Ephemeral propositions might be invalidated at some point, and are
therefore not duplicable. On the other hand, persistent propositions are, once
valid, always valid.
The persistently modality $\always$ describes the persistent part of the proposition $P$.
The later modality $\later$ is mostly technical.
It roughly means that the proposition $P$ holds \emph{after one logical step of execution}.
The invariant $\knowInv{}{P}$ states that $P$ holds for now and forever.
$P$ can be accessed for one step of execution, but must be restored by the end of the step.
We refer the reader to \citet{iris} for details.

To connect the logic with the machine's physical state, the program logic contains points-to separation logic predicates for registers ($r \rmapsto{} \V{w}$), memory addresses ($a \mapsto \V{w}$) and contiguous memory ranges ($\crange{b}{e} \mapsto l$ for a list of words $l$).
Furthermore, two different kinds of triples (both built on top of Iris' weakest-precondition) describe program specifications:
\begin{itemize}
  \item The single instruction execution $\stepspecbase{P}{s.\, Q}$ states that,
    starting in a state satisfying $P$, the machine reaches a state satisfying $Q$ after one step of execution, with execution state $s$.
    We omit $s$ if the machine remains in the $\Running$ state.

  \item The complete, safe execution $\contspecbase{P}{}$ states that,
        starting in a state satisfying $P$, the machine either diverges, or runs until it halts or fails without breaking registered invariants.
\end{itemize}
Each instruction comes with a program logic rule describing its behavior in the logic.
For example, \Cref{fig:loadwprules} shows the $\instr{load}$ instruction rule.
We refer to \citet{cerise} for a more detailed exposition.
\begin{figure}
  \begin{mathpar}
    \centering
    \inferH{} 
    {\X{ValidPC}(\ppc,\bpc,\epc,\apc) \\
      p \in \{\perm{RO}, \perm{RX}, \perm{RW}, \perm{RWX}\} \\
      b \leq a < e \\
      \X{decode}(n) = \instr{Load}\; \I{r_{dst}}\; \I{r_{src}}
    }
    {\stepspecbase
      {
        \begin{array}{l}
          \pc \rmapsto (\ppc,\bpc,\epc,\apc) \\
          {} \ast \apc \mapsto n
           \ast a \amapsto{} w \\
          {} \ast \I{r_{src}} \rmapsto (p,b,e,a)
          \ast \I{r_{dst}} \rmapsto - \\
        \end{array}
      }
      {
        \begin{array}{l}
          \pc \rmapsto (\ppc,\bpc,\epc,\diff{\apc + 1}) \\
          {} \ast \apc \mapsto n
          \ast a \mapsto w \\
          {} \ast \I{r_{src}} \rmapsto (p,b,e,a)
          \ast \I{r_{dst}} \rmapsto \diff{w} \\
        \end{array}
      }
    }
  \end{mathpar}%
  \begin{align*}
    \text{where}~\X{ValidPC}(p, b, e, a, v) &\triangleq p \in \{\perm{RX}, \perm{RWX}\} \land b \le a < e
  \end{align*}

\caption{WP rule for a valid $\instr{load}$.
Changes between the pre- and postcondition are marked in \diff{orange}.}
\label{fig:loadwprules}
\end{figure}



\cerise{} establishes \emph{capability safety}, a formal property that captures interactions between known, trusted code $K$ invoking  (\Cref{fig:interaction_scenario:calling_adv}) and being invoked by (\Cref{fig:interaction_scenario:being_called_by_adv}) unknown, untrusted code $U$.
Safety of such interactions is proved by establishing a Hoare triple for the known code using
program logic rules, and appealing to capability safety when unknown code is invoked.
Capability safety is formalized as a universal contract: a contract for arbitrary code, expressing that its execution will not exceed the authority of the capabilities it has access to, or break registered invariants.
The authority of capabilities is defined using a logical relation\footnote{Traditionally, logical relations are used to interpret type systems,
  but they can also be used to interpret the universal type used by an untyped
  language \cite{DBLP:conf/dagstuhl/Pitts10,DBLP:conf/eurosp/DevrieseBP16},
  which is essentially what is happening in Cerise.} $\safeVl$, which is defined using the program logic, and which we explain in the following.

\begin{figure}
\raggedright
\begin{align*}
\begin{array}[t]{ll}
\safeE{w} \triangleq~ &
    \setlength\arraycolsep{2pt}
    \begin{array}[t]{l}
      \contspecbase
        {\begin{array}{l}
          \pc \rmapsto w \ast
          \left(
          \displaystyle\Sep_{
          \substack{(r, w_{r}) \in \V{regs}}
          }
          r \rmapsto w_{r} \ast \safeV{w_{r}} \right)
        \end{array}
        }{}
    \end{array}
  \\[2.8em]
\safeV{w} \triangleq~
          &
            \left\{
  \arraycolsep=2.2pt\def\arraystretch{1.1}
  \begin{array}[c]{LCL}
    \rowstyle{\color{black}}
    \safeV{z} & \triangleq & \TRUE ~\text{for}~ z \in \mathbb{Z}
  \\
    \safeV{\perm{o},-,-,-} & \triangleq & \TRUE
    \\
    \safeV{\enter,b,e,a}  & \triangleq & \later\,\square\,\safeE{\RX,b,e,a}
  \\
    \safeV{\RW/\RWX,b,e,-} & \triangleq &
                                          \Sep_{a \in \crange{b}{e}} \knowInv{a}{\exists w \ldotp a \mapsto w \ast \safeV{w}
                                          }
    \\
    \safeV{\RO/\RX,b,e,-} & \triangleq &
       \Sep_{a \in \crange{b}{e}}
                                         \exists P \ldotp \knowInv{a}{\exists w \ldotp a \mapsto w \ast P(w)
                                         }
    \\ & & \quad\quad\quad\quad  {} \ast \later\square\left(\forall w \ldotp P(w) \wand \mathcal{V}(w)\right)
    \\
    \rowstyle{\color{newCerisierColor}}
    \safeV{[\V{ps},\otype_{\X{b}},\otype_{\X{e}},\otype_{\X{i}}]} & \eqdef
                                & (\pure{\perm{s} \in \V{ps}}
                                  \wand \safeseal(\otype_{\X{b}},\otype_{\X{e}})) \\
                   & & {} \ast (\pure{\perm{U} \in \V{ps}}
                       \wand \safeunseal(\otype_{\X{b}},\otype_{\X{e}})) \\
                   & & {} \ast (\pure{\V{ps} = \perm{SU}}
                       \wand \safeattest(\otype_{\X{b}},\otype_{\X{e}}))
    \\
    \rowstyle{\color{newCerisierColor}}
    \safeV{\sealcap} & \eqdef & \exists P \ldotp \sealPred(\otype,P) \ast P(\V{sc})
  \end{array}
            \right.
\end{array}
\end{align*}

where
\begin{align*}
  \begin{array}[t]{LCL}
    \rowstyle{\color{black}}
    \V{regs} & \eqdef & {\left\{
                        (r, w) \; \left | \; \right.
                        (r, w) \in \V{reg},\ r \neq\pc
                        \right\}}
    \\
    \rowstyle{\color{newCerisierColor}}
    \safeseal(\otype_{\X{b}},\otype_{\X{e}}) & \eqdef
                      & \Sep_{\otype_a \in \crange{\otype_b}{\otype_e}}
                        \exists P \ldotp
                        \pure{\forall w \ldotp \X{Persistent}(P(w))}
                        \ast \sealPred(\otype_a)(P)
                       \\ & & \quad\quad\quad\quad\quad {} \ast (\forall w \ldotp \safeV{w} \wand P(w))
    \\
    \safeunseal(\otype_{\X{b}},\otype_{\X{e}}) & \eqdef
                      & \Sep_{\otype_a \in \crange{\otype_b}{\otype_e}}
                        \exists P \ldotp \sealPred(\otype_a)(P) \ast (\forall w \ldotp P(w) \wand \safeV{w})
    \\
    \safeattest(\otype_{\X{b}},\otype_{\X{e}}) & \eqdef
                      & \Sep_{\otype_a \in \crange{\otype_b}{\otype_e}}
                        \exists \V{tidx} \ldotp
                        \pure{ \tidxofot(\otype_a) = \V{tidx} }
                        \\ && \quad\quad\quad\quad\quad
                        {} \ast \knowInv{\V{tidx}}{\exists I \ldotp \enclaveLive{tidx}{I} \lor \enclavePrev{tidx} }
\end{array}
\end{align*}
\caption{Logical relation defining ``safe to share'' ($\Vl$) and ``safe to execute'' ($\El$).
  The definitions marked in \newCerisier{blue} can be ignored for now, and will be explained in
  \Cref{sec:seal_caps}. $\safeattest$ is included for completeness and necessary for deinitialization
  of enclaves, but is not discussed further.
}
\label{fig:logrel}
\end{figure}


\begin{theorem}[\cerise{} Universal Contract]
  \label{thm:cerise-univ-contract}
  \begin{equation*}
    \proves
    \contspecbase
    {\begin{array}{l}
      \displaystyle\Sep_{
      \substack{r \in \RegName}
      }
      \left(\exists w\ldotp r \rmapsto w \ast \safeV{w}\right)
    \end{array}
  }{}
\end{equation*}
\end{theorem}

Note that the \cerise{} universal contract is usually phrased in a different form, as a Fundamental Theorem of Logical Relations (FTLR) for $\safeVl$.
This FTLR is equivalent to the universal contract in Theorem~\ref{thm:cerise-univ-contract} but harder to explain, so we will not discuss it in this paper.

\Cref{fig:logrel} shows the definition of $\safeVl$ (``safe to share (with the adversary)'') and the supporting notion $\safeEl$ (``safe to execute (by the adversary)'').
Intuitively, a word $w$ is \emph{safe to share} if authority is available for anything the adversary can do with the word.
We say that a word $w$ is \emph{safe to execute} if the machine executes safely when $w$ is installed in the $\pc$ with safe words in all registers.
Integers $z \in \ZZ$ and null-permission ($\perm{O}$) capabilities are trivially safe.

A sentry capability ($\perm{E}$) is safe to share when the code it encapsulates is safe to execute.
The later modality $\later$ ensures that the apparently circular definition of $\safeVl$ is well-defined.
The persistent modality $\always$ expresses the fact that the sentry capability can be invoked several times.
A $\RW / \RWX$\footnote{
  The universal contract does not distinguish between $\RW / \RWX$,
  because it models an interpretation of capability safety where
  the security model makes no restrictions on the code blocks that an adversary
  has access to for execution.
}
capability gives access to the memory region $\crange{b}{e}$, and the (authority of the) words it contains.
They are safe to share only if the memory locations are owned by an invariant that also enforces safety of the words they contain.
Finally, similarly to $\RW / \RWX$, a $\RO / \RX$ capability gives access to the memory region in between its bounds, but its contents cannot be modified.
The logical relation expresses read-only access by requiring addressed memory to invariantly satisfy an arbitrary predicate $P$, that is only known to entail $\safeVl$.

Finally, \Cref{thm:cerise-univ-contract} expresses that if all registers contain safe to share values, the machine executes safely.
It provides the foundation for reasoning about invoking untrusted code from trusted code, where verifying the safety of register contents becomes a proof obligation.

\subsection{Sealed Capabilities}
\label{sec:seal_caps}

\begin{figure}
  \small
{
\begin{minipage}[t]{0.45\textwidth}
\[
\arraycolsep=3pt
\begin{array}{LCLCL}
  \rowstyle{\color{newCerisierColor}}
  \otype & \in & \OType & \eqdef & [0, \OTypeMax] \\
  \sealperm & \in & \SealPerm & \mathrel{::=} & \perm{SU} \mid \perm{S} \mid \perm{U} \mid \perm{o} \\
  \V{sr} & \in & \SealRange & \eqdef & \{
                                      \left[\V{sp}, \otype_{b}, \otype_{e}, \otype_{a}\right]
                                       \mid \\
         &     &             &        & \quad \otype_{b}, \otype_{e}, \otype_{a} \in \OType \} \\
\end{array}
\]
\end{minipage}%
\begin{minipage}[t]{0.48\textwidth}
\[
\arraycolsep=3pt
\begin{array}{LCLCL}
  \rowstyle{\color{newCerisierColor}}
  \V{sc} & \in & \Sealable & \eqdef & \CCap \mathbin{+} \SealRange \\
  \rowstyle{\color{newCerisierColor}}
  \V{sdc} & \in & \SealedCap & \eqdef & \OType \times \Sealable \\
  \rowstyle{\color{black}}
  w & \in & \Word & \eqdef & \ZZ + \cancel{\CCap}~ \newCerisier{\Sealable \mathbin{+}~\SealedCap} \\
  \rowstyle{\color{black}}
\end{array}
\]
\end{minipage}%
}
\[
\begin{array}{@{}lcl@{}}
  i \in \Instruction & \mathrel{::=} &\cdots \mid
  \begin{array}[t]{@{}+l@{}}
  \rowstyle{\color{newCerisierColor}}
               \instr{cseal}\; r_\mathrm{d}\; r_\mathrm{1}\; r_\mathrm{2} \mid
               \instr{cunseal}\; r_\mathrm{d}\; r_\mathrm{1}\; r_\mathrm{2} \mid
  \rowstyle{\color{black}}
  \end{array}
\end{array}
\]

  \footnotesize
  \noindent\makebox[\textwidth]{\begin{tabular}{|C|L|L|}
    \hline
  $i$ & $\instrsem{i}(\confv)$ & Conditions \\ \hline
  \rowstyle{\color{newCerisierColor}}
  $\instr{cseal}\; r_{d} \; r_{1} \; r_{2}$
  & $\X{updPC}(\confv[\X{reg}.r_{d} \mapsto \sealed{\otype_{a}}{\V{sc}}])$
  & \!\!\!\begin{tabular}{l}
    $\confv.\X{reg}(r_1) = [\V{sp},\otype_{b},\otype_{e},\otype_{a}]$
    and $\confv.\X{reg}(r_2) = \V{sc}$ \\
    and $\V{sp} \in \{\perm{S}, \perm{SU}\}$
    and $\otype_{b} \le \otype_{a} < \otype_{e}$ \\
    \end{tabular}
  \\ \hline
  $\instr{cunseal}\; r_{d} \; r_{1} \; r_{2}$
  & $\X{updPC}(\confv[\X{reg}.r_{d} \mapsto \V{sc}])$
  & \!\!\!\begin{tabular}{l}
    $\confv.\X{reg}(r_1) = [\V{sp},\otype_{b},\otype_{e},\otype_{a}]$
    and $\V{sp} \in \{\perm{U}, \perm{SU}\}$ \\
    and $\confv.\X{reg}(r_2) = \sealed{\otype_{a}}{\V{sc}}$
    and $\otype_{b} \le \otype_{a} < \otype_{e}$
    \end{tabular}
  \\ \hline
  \end{tabular}
  \rowstyle{\color{black}}
}
\caption{\label{fig:opsem_excerpt_sealing}
Excerpt of the operational semantics with capability sealing extensions, extending \Cref{fig:opsem_excerpt}.
}
\end{figure}


The idea of \cheri{} sealed capabilities dates back to \citet{morris_protection_1973} and implements a form of symbolic cryptography \citep{sumiiLogRelEncryption}, orthogonal to trusted execution.
The act of sealing makes a capability opaque and immutable: it cannot be modified or dereferenced until it is unsealed.
To support sealing, Cerisier extends Cerise operational semantics with two instructions,
$\instr{cseal}$ and $\instr{cunseal}$ to perform sealing and unsealing of capabilities.
The changes to the operational semantics are shown in \Cref{fig:opsem_excerpt_sealing}.

To seal a capability, the machine takes a capability $\V{sc}$
and a valid \emph{sealing capability} $\sealrange$,
a new type of capability holding the authority to seal (or unseal).
The permission $\V{sp}$ describes the sealing permission, for either sealing ($\perm{S}$) or unsealing ($\perm{U}$).
Instead of addresses, $\crange{\otype_{b}}{\otype_{e}}$ describes a range of so-called \emph{otypes},
acting as key identifiers.
Similarly to regular capabilities, the current otype $\otype_{a}$ needs to be in bounds to be used.
The action of sealing creates a \emph{sealed capability} $\sealed{\otype_{a}}{\V{sc}}$,
making it opaque and immutable.
Conversely, to unseal a capability, the machine takes a sealed capability $\sealed{\otype_{a}}{\V{sc}}$
and a valid (un-)sealing capability $\sealrange$ with unsealing permission ($\perm{U} \in \V{sp}$)
and matching otype, and retrieves the initial capability $\V{sc}$.

Sealing capabilities cannot be forged, they can only be monotonically derived from other sealing capabilities.
In Cerisier, new sealing capabilities, with fresh unique otypes, are generated when enclaves are initialized, as we detail in the next section.
The operational semantics also includes a couple of other new instructions: $\instr{getwtype}$ for querying the type of a word (integer, regular capability, sealing capability or sealed capability) and $\instr{getotype}$ for querying the otype $\otype$ of a sealed capability; their precise semantics can be found in the accompanying Rocq formalization.

By sharing restricted sealing capabilities with an adversary, one obtains a
form of asymmetric cryptography: A component holding a private sealing
capability $[\perm{SU}, \otype_{b}, \otype_{e}, \otype_{a}]$, can, e.g.,
disclose the unsealing authority
$[\perm{U}, \otype_{b}, \otype_{e}, \otype_{a}]$ to the adversary and keep the
full authority private to obtain a form of cryptographic signatures.
Conversely, making the sealing authority public, enables a form
of asymmetric encryption. In fact, the public/private key pair for attestation
suggested in \Cref{fig:secure_outsourced_computation_pseudocode}, corresponds in
reality to a single $\perm{SU}$ sealing capability used in this way for signing
messages.

Following \citet{stkTokens} we extend the Cerise program logic with support for sealing by adapting a technique by \citet{sumiiLogRelEncryption}.
The logical relation is extended by the definition marked in \textcolor{newCerisierColor}{blue} in \Cref{fig:logrel}.
Logically, we associate a value predicate $P$ with every otype $o$, such that, for any capability sealed with this otype, the predicate holds.
For instance, a possible predicate could be ``any capability sealed with otype $\otype$ will have its current address set to 42''.
The association is witnessed by a predicate $\sealPred(\otype,P)$.
When sharing $\otype$-sealed values with the adversary, there is an obligation to show that the
value being sealed respects $P$
; and dually, it gives the guarantee that $\otype$-sealed values received from the adversary also
respect $P$.
We defer further details to Supplemental Materials, under \appendixReplace{\Cref{app:sealing}}{E}.

\subsection{Exclusive Access and Memory Revocation}
\label{sec:mem_revocation}

As mentioned in \Cref{sec:introduction}, when an enclave is initialized, Cerisier implements the exclusive memory
ownership primitive of CHERI-TrEE by a memory sweep, to check that
the entire memory (as well as the register file) does not contain a
capability overlapping the new enclave's memory.
We model the memory sweep as described by the \cheritree{}'s operational semantics,
and we discuss its hardware implementation in \Cref{sec:discussion_memory-sweep}.

\begin{figure}
  \small

\[
\arraycolsep=3pt
\begin{array}{@{}lcl@{}}
  i \in \Instruction & \mathrel{::=} &\cdots \mid
  \begin{array}[t]{@{}+l@{}}
  \rowstyle{\color{newCerisierColor}}
    \instr{isunique}\; r_\mathrm{d}\; r_\mathrm{s} \mid
  \rowstyle{\color{black}}
  \end{array}
\end{array}
\]
  \footnotesize
  \noindent\makebox[\textwidth]{\begin{tabular}{|C|L|L|}
    \hline
  $i$ & $\instrsem{i}(\confv)$ & Conditions \\ \hline
  \rowstyle{\color{newCerisierColor}}
  $\instr{isunique}\; r_{d}\; r_{s}$
  & \!\!\!\begin{tabular}{l}$\X{updPC}(\confv[\X{reg}.r_{d} \mapsto z])$\end{tabular}
  & \!\!\!\begin{tabular}{l}
    $\left(
    \begin{tabular}{c}
      $\confv.\X{reg}(r_{s}) = (p,b,e,a)$ \\
      $\lor$
      $\confv.\X{reg}(r_{s}) = \{(p,b,e,a)\}_{\otype_{i}}$
    \end{tabular}
    \right)
    $\\
    and
    $\left(
    \begin{tabular}{ll}
      if & $\sweepr(\confv)(r_{s})$\\
      then & $z=1$ \quad else $z=0$
    \end{tabular}
    \right)$
  \end{tabular}
  \\ \hline
  \end{tabular}
  \rowstyle{\color{black}}
}

\footnotesize
\begin{mathpar}
  \color{newCerisierColor}
  \overlap(w_1)(w_2) =
  \left\{
  \begin{array}{ll}
    \crange{b_1}{e_1} \cap  \crange{b_2}{e_2}
    &\begin{array}{l}
      \text{if }
      \left(w_1 = (p_1,b_1,e_1,a_1) \lor w_1 = \sealed{\otype}{(p_1,b_1,e_1,a_1)}\right)
      \\ \text{and}
      \left(w_2 = (p_2,b_2,e_2,a_2) \lor w_2 = \sealed{\otype}{(p_2,b_2,e_2,a_2)}\right)
    \end{array}
    \\ \bottom & \text{otherwise}
  \end{array}
  \right.

  \sweepr(\confv)(r_s) =
  \begin{array}{l}
    \forall r \in \V{dom}(\confv.\X{reg} \setminus r_s).\;
    \neg \overlap(\confv.\X{reg}(r), \confv.\X{reg}(r_s))\\
     {} \land \forall a \in \V{dom}(\confv.\X{mem}).\;
    \neg \overlap(\confv.\X{mem}(a), \confv.\X{reg}(r_s))
  \end{array}

  \color{black}

\end{mathpar}
\caption{\label{fig:opsem_excerpt_extension_revocation}
Excerpt of the operational semantics with revocation-related extensions, extending \Cref{fig:opsem_excerpt}.
}
\end{figure}


%
In addition to enclave initialization,
the sweep can also be used after initialization through an $\instr{isunique}$
instruction, in order to support dynamically growing enclaves and memory-sharing
communication. The operational semantics is shown in
\Cref{fig:opsem_excerpt_extension_revocation}, by comparison to
\Cref{fig:opsem_grammar} and \Cref{fig:opsem_excerpt}. The instruction
$\instr{isunique}\; r_{d}\; r_{s}$ requires a (possibly sealed) capability
$(p,b,e,a)$ in source register $r_{s}$. The machine will sweep its entire state
$\confv$, i.e., the register file (excluding $r_{s}$) and the entire memory, to
check that no capabilities overlap the source capability. The result is stored
in destination register $r_{d}$ as an integer $1$ or $0$, indicating that the
source capability is unique, or has duplicates, respectively.

Logically, the memory sweep, akin to a stop-the-world garbage collection,
corresponds to a \emph{revocation} of ownership:
memory that was previously owned by the adversary
and that overlaps with the memory region of the capability in $r_{s}$
is logically revoked and reallocated to the enclave.
The revocation step is challenging to support in
Cerisier, because it is a highly non-modular operation: the revocation affects the
global state, whereas separation logics like Iris, which \cerisier{} builds
upon, support only local (so-called frame-preserving) operations, in order to
validate the celebrated frame rule of separation logic. Fortunately, we can
adapt and reuse a technique by \citet{hurSeparationLogicPresence2011} for
reasoning about garbage collection, which involves a similar global operation.
They introduce an additional layer of \emph{logical} memory, which allows
addresses to be reallocated when the unique memory ownership sweep succeeds. In
essence, the logical memory decorates the physical memory with version numbers,
which are incremented at each successful memory sweep operation. An invariant
links the physical state of the machine to logical memory, ensuring that only
the latest version number of each memory address is reachable from garbage
collection roots.

Because our support for revocation is an adaptation of a known technique to
our setting, we omit full details and simply summarize the program logic
interface for revocation. Essentially, we add versions to points-to predicates:
$a \amapsto{v} w$ expressing that address $a$ at version $v$ points to logical
word $\V{w}$. \emph{Logical words} are either integers or logical capabilities
$\lcap{p}{b}{e}{a}{v}$, which carry a version number (note that all addressed logical
addresses must be at the same version). Logical capabilities can only be used
with the points-to predicate for the same version. Existing \cerise{} program
logic rules are straightforwardly extended with versions, which they leave
unmodified.

\begin{figure}
  \begin{mathpar}
    \centering
    \inferH{IsUniqueValid}
    {\X{ValidPC}\lcap{\ppc}{\bpc}{\epc}{\apc}{\vpc} \\
      p \in \{\perm{RO}, \perm{RX}, \perm{RW}, \perm{RWX}\} \\
      \X{decode}(n) = \instr{isunique}\; \I{r_{dst}}\; \I{r_{src}}\\
      \V{a_{pc}} \notin \crange{b}{e}
    }
    {
      \stepspecbase{
        \begin{array}{l}
          \pc \rmapsto \lcap{\ppc}{\bpc}{\epc}{\apc}{\vpc} \\
          {} \ast \apc \amapsto{\vpc} n \\
          {} \ast \I{r_{src}} \rmapsto \lcap{p}{b}{e}{a}{v}\\
          {} \ast \I{r_{dst}} \rmapsto -
        \end{array}
      }
      {
        \begin{array}{l}
          \pc \rmapsto \lcap{\ppc}{\bpc}{\epc}{\emphc{\apc+1}}{\vpc}
          {} \ast \apc \amapsto{\vpc} n \\
          {} \ast \left(
          \begin{array}{lc}
            &\I{r_{dst}} \rmapsto \emphc{0}
            \ast \I{r_{src}} \rmapsto \lcap{p}{b}{e}{a}{v}\\
            \lor&
            \I{r_{dst}} \rmapsto \emphc{1}
            {} \ast \I{r_{src}} \rmapsto \lcap{p}{b}{e}{a}{\emphc{v+1}}\\
            &\emphc{{} \ast \crange{b}{e} \amapsto{v+1} \V{lw}}
          \end{array}
          \right)
        \end{array}
      }
    }
\end{mathpar}

\caption{WP rule for a valid case of the \instr{isunique} instruction with distinct PC, source, and destination registers. The postcondition asserts that either the sweep was unsuccessful, writing 0 to the destination register, or it was successful and we increment the version number of the capability and all addresses within its bounds.
}
\label{fig:isuniquewprules_full}
\end{figure}


%
The program logic rule for the
$\instr{isunique}$ instruction in Figure~\ref{fig:isuniquewprules_full} shows
the interface for memory revocation. Assuming we have the points-to predicate
for the entire logical memory region of the swept capability, the rule describes
what happens in two cases. In case of a successful sweep (\ie~the second part of
the disjunction), the logical version of the capability in the source register
is incremented, and we obtain a fresh points-to predicate for the updated memory
region. The points-to predicates of the old version may still exist, but are not
linked to the physical state anymore. The soundness proof of the rule requires justifying
that bumping the version number of the swept memory region does not invalidate
the invariant linking physical and logical state. Informally, this is the case
because the physical sweep ensures that there is no other capability overlapping
the memory region, and so no other capabilities have to be updated.

\subsection{Secure Hashing Primitives}
\label{sec:secure-hash-prim}

\begin{figure}
  \small

\[
\arraycolsep=3pt
\begin{array}{@{}lcl@{}}
  i \in \Instruction & \mathrel{::=} &\cdots \mid
  \begin{array}[t]{@{}+l@{}}
               \instr{hash}\; r_\mathrm{d}\; r_\mathrm{s} \mid
               \instr{hashconcat}\; r_\mathrm{d} \; \rho_\mathrm{1} \; \rho_\mathrm{2}
  \end{array}
\end{array}
\]

  \footnotesize
  \noindent\makebox[\textwidth]{\begin{tabular}{|C|L|L|}
    \hline
  $i$ & $\instrsem{i}(\confv)$ & Conditions \\ \hline
  \rowstyle{\color{newCerisierColor}}
  $\instr{hash}\; r_{d}\; r_{s}$
  & \!\!\!\begin{tabular}{l}$\X{updPC}(\confv[\X{reg}.r_{d} \mapsto \hash(w)])$\end{tabular}
  & \!\!\!\begin{tabular}{l} $\confv.\X{reg}(r_{s}) = w$ \end{tabular}
  \\ \hline
  $\instr{hashconcat}\; r \; \rho_1 \; \rho_2$
  & $\X{updPC}(\confv[\X{reg}.r \mapsto z])$
  & \!\!\!\begin{tabular}{l}
      for $i \in \{1, 2\}$, $z_i = \X{getWord}(\confv, \rho_i)$ \\
      and $z_i \in \ZZ$ and $z = z_1 \hashconcat z_2$
    \end{tabular}
  \rowstyle{\color{black}}
  \\ \hline
  \end{tabular}
  \rowstyle{\color{black}}
}

\begin{mathpar}
  \X{getWord}(\confv, \rho) = \left\{
    \begin{array}{ll}
      \rho & \text{if } \rho \in \ZZ \\
      \confv.\X{reg}(\rho) & \text{if } \rho \in \X{RegName}
    \end{array}
  \right.
\end{mathpar}
\caption{\label{fig:opsem_excerpt_extension_hashing}
Excerpt of the operational semantics containing secure hashing extensions and assumptions, extending \Cref{fig:opsem_excerpt}.
The notation $\hashconcat$ is the machine primitive of hash concatenation.
}
\end{figure}


As explained in the Introduction, attestation often involves hashing of a component's source code, etc.
Therefore Cerisier also includes primitives for hashing, and we use them
in \Cref{sec:case_study_mutual_attest} for implementing a mutual attestation scheme.
The operational semantics of the hashing instructions is shown in \Cref{fig:opsem_excerpt_extension_hashing}.
The $\instr{hash}\; r_{d}\; r_{s}$ instruction stores the hash of the word in $r_{s}$ in $r_{d}$ and
$\instr{hashconcat}$ takes two hash values $z_{1}$ and $z_{2}$, and computes their concatenation $z_{1} \hashconcat z_{2}$.

The hash computation is intended to model a secure hash algorithm, which we
axiomatize in the form of injective and collision-free operations on natural
numbers. These assumptions are not contradictory because \cerisier{} allows
unbounded integers, and it allow us to avoid cross-cutting but orthogonal
changes that would be required for a symbolic or probabilistic model. The
precise assumptions we use are:
\[
  \begin{array}{lcl}
    \hash(\V{mem_{1}}) = \hash(\V{mem_{2}}) & \implies & \V{mem_{1}} = \V{mem_{2}} \\
    (z_{1} \hashconcat z_{2}) = (z_{1}' \hashconcat z_{2}') & \implies & z_{1} = z_{1}' \land z_{2} = z_{2}'.
  \end{array}
\]
In practice, these conditions ensure that when a statically computed enclave identity
(\Cref{fig:secure_outsourced_computation_pseudocode}, $\mathrm{ID}_{\mathrm{soc}}$)
matches a dynamically computed enclave identity
(\Cref{fig:secure_outsourced_computation_pseudocode}, $k_p$),
we may conclude that the hashed memory regions must be equal.

\section{Attestation}
\label{sec:attestation}

In this section we detail the operational semantics of the new attestation-related instructions
included in Cerisier. The instructions are inspired by similar ones in \cheritree{}.
Next, we show how the secure outsourced computation example from \Cref{sec:soc}
is implemented in Cerisier, before we turn to reasoning about attestation and
clients of attested enclaves in the program logic.
Finally, we discuss how to reason about unknown adversarial code using the novel
Cerisier universal contract, the proof of which is outlined in the final subsection.

\subsection{Operational Semantics}
\label{sec:opsem}
\begin{figure}
  \small
{
\setlength{\abovedisplayskip}{0pt}%
\setlength{\abovedisplayshortskip}{0pt}%
\setlength{\belowdisplayskip}{0pt}%
\begin{minipage}[t]{0.35\textwidth}
\[
\arraycolsep=3pt
\begin{array}{LCLCL}
  \rowstyle{\color{newCerisierColor}}
  \V{ecn} & \in & \ENum & \eqdef & \NN \\
  \V{I} & \in & \EId & \eqdef & \NN \\
  \V{tidx} & \in & \TIndex & \eqdef & \NN \\
\end{array}
\]
\end{minipage}%
\begin{minipage}[t]{0.55\textwidth}
\[
\arraycolsep=3pt
\begin{array}{LCLCL}
  \rowstyle{\color{newCerisierColor}}
  \V{etbl} & \in & \ETable & \eqdef & \TIndex \rightarrow  \EId \\
  \rowstyle{\color{black}}
  \confv & \in & \ExecConf & \eqdef & \RegFile \times \Mem  \newCerisier{~\times~\ETable \times \ENum}
\end{array}
\]
\end{minipage}%
}
\[
\begin{array}{@{}lcl@{}}
  i \in \Instruction & \mathrel{::=} &\cdots \mid
  \begin{array}[t]{@{}+l@{}}
  \rowstyle{\color{newCerisierColor}}
               \instr{einit}\; r_\mathrm{s} \mid
               \instr{edeinit}\; r_\mathrm{s} \mid
               \instr{estoreid}\; r_\mathrm{d}\; r_\mathrm{1}\; r_\mathrm{2} \mid
               \instr{hash}\; r_\mathrm{d}\; r_\mathrm{s} \mid
               \instr{hashconcat}\; r_\mathrm{d} \; \rho_\mathrm{1} \; \rho_\mathrm{2}
  \rowstyle{\color{black}}
  \end{array}
\end{array}
\]

  \footnotesize
  \noindent\makebox[\textwidth]{\begin{tabular}{|C|L|L|}
    \hline
  $i$ & $\instrsem{i}(\confv)$ & Conditions \\ \hline
  \rowstyle{\color{newCerisierColor}}
  %
  $\instr{einit}\; r_{1}\; r_{2}$
  & \!\!\!\begin{tabular}{l}
    {$\X{updPC}(\confv[$} \\
    \quad{$\X{mem}.b \mapsto (\RW,b',e',a')$}, \\
    \quad{$\X{mem}.b' \mapsto [\perm{su},\otype_{a},\otype_{a+2},\otype_{a}]$}, \\
    \quad{$\X{etbl}.\V{tidx} \mapsto I$}, \\
    \quad{$\X{EC} \mapsto \confv.\X{EC} + 1$}, \\
    \quad{$\X{reg}.r_{1} \mapsto (\enter, b, e, b+1)$}, \\
    \quad{$\X{reg}.r_{2} \mapsto 0$])} \\
  \end{tabular}

  & \!\!\!\begin{tabular}{l}
    $r_{1} \neq \pc $
    and $\confv.\X{reg}(r_{1}) = (\RX,b,e,a)$
    and $b < e$ \\
    and $\confv.\X{reg}(r_{2}) = (\RW,b',e',a')$
    and $b' < e'$ \\
    and $\sweepr(\confv)(r_{1})$
    and $\sweepr(\confv)(r_{2})$ \\
    and $\left(\forall l \in \crange{b+1}{e}.\;
    \confv.\X{mem}(l) \in \ZZ\right)$\\
    and $I = \hash(b) \hashconcat \hash(\confv.\X{mem}(\; \crange{b+1}{e} \; ))$ \\
    and $\V{tidx} = \freshtidx(\confv)$
    and $\otype_{a} = \confv.\X{EC}*2$\\
    \end{tabular}
  \\ \hline
  $\instr{edeinit}\; r$
  & \!\!\!\begin{tabular}{l}$\X{updPC}(\confv[\X{etbl}.\X{tidx} \mapsto \emptyset ])$ \end{tabular}
  & \!\!\!\begin{tabular}{l}
    $\confv.\X{reg}(r) = [\perm{su},\otype_{a},\otype_{a+2},\_]$ \\
    and $\X{tidx} = \tidxofot(\otype_{a})$\\
    and $\confv.\X{etbl}(tidx) = I$\\
    \end{tabular}
  \\ \hline
  $\instr{estoreid}\; r_d \; r_s$
  & \!\!\!\begin{tabular}{l}$\X{updPC}(\confv[\X{reg}.r_{d} \mapsto I ])$ \end{tabular}

  & \!\!\!\begin{tabular}{l}
    $\confv.\X{reg}(r_{s}) = \otype_{s}$
    and $\X{tidx} = \tidxofot(\otype_{s})$\\
    and $\confv.\X{etbl}(tidx) = I$\\
    \end{tabular}
  \\ \hline
  \end{tabular}
  \rowstyle{\color{black}}
}

\footnotesize
\begin{mathpar}
  \color{newCerisierColor}
  \tidxofot(\otype_a) = \left\{
    \begin{array}{ll}
      \otype_a \slash 2 & \text{if $\V{is\_even}(\otype_a)$}\\
      (\otype_a-1) \slash 2 & \text{otherwise}
    \end{array}
  \right.

  \freshtidx(\confv) = \confv.\X{EC}
  \color{black}
\end{mathpar}
  
\caption{\label{fig:opsem_excerpt_extension_attestation}
Excerpt of the operational semantics with attestation-related extensions, extending \Cref{fig:opsem_excerpt}.
}
\end{figure}


The new attestation-related instructions in the \cerisier{} ISA are $\instr{einit}$, $\instr{estoreid}$, and $\instr{edeinit}$ for enclave initialization, attestation, and deinitialization, respectively.
Their operational semantics are defined in \Cref{fig:opsem_excerpt_extension_attestation}.
They interact with an \emph{enclave table} that tracks registered enclaves’ identities.
A special register $\sreg$, not directly accessible by the user, always contains the next free
available table index.
We first discuss the $\instr{einit}$ instruction, whose effects on the machine state are depicted in \Cref{fig:einit}.

\begin{figure}[ht]
  \centering

  \newcommand{\haddr}{0.5}
  \newcommand{\hcontent}{1}
  \newcommand{\wcontent}{2.5}
  \newcommand{\wpadding}{0.5}
  \newcommand{\wpaddingtext}{0.25}
  \newcommand{\wpaddingdata}{1}
  \newcommand{\wdata}{\wpadding}

  \newcommand{\hpaddingcode}{1}
  \newcommand{\hdata}{0}
  \newcommand{\hcode}{\hdata+\hcontent+\hpaddingcode}

  \begin{subfigure}{0.45\columnwidth}
    \includegraphics[width=\columnwidth]{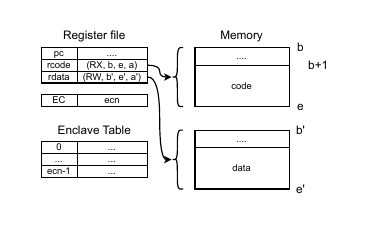}

\caption{\label{fig:einit_before}
  Cerisier machine state prior to \instr{einit}}
  \end{subfigure}
  \begin{subfigure}{0.45\columnwidth}
    \includegraphics[width=\columnwidth]{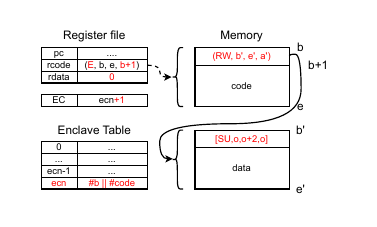}

\caption{\label{fig:einit_after}
  Cerisier machine state following \instr{einit}}
  \end{subfigure}

\caption{\label{fig:einit}
  Representation of the effects in memory of the \instr{einit} instruction. }
\end{figure}

Enclave initialization using the $\instr{einit}\; \reg{1}\; \reg{2}$ instruction expects the enclave to consist of two parts: a code and data region.
As \Cref{fig:einit} shows, $\instr{einit}$ requires a non-empty \perm{RW} capability for the data region $\crange{b'}{e'}$, and a non-empty $\perm{RX}$ capability for the code region at $\crange{b}{e}$.
It will check unique ownership of the two regions, similar to what was done for \instr{isunique} in \Cref{sec:secure-hash-prim}.
Following successful sweeps, enclave initialization proceeds without further validation.

The enclave's memory layout is updated as follows.
The data region capability $(\RW, b', e', a')$ is stored at the code region's base address $b$.
The rest of the code region $\crange{b+1}{e}$ is expected to contain the enclave instructions.
A sentry capability $(\enter, b, e, b+1)$ to the first of these instructions will be stored in $\reg{1}$, to be safely invoked by the context.
Two otypes $\otype_{a} := 2*\V{tidx}$ and $\otype_{a}+1$ are reserved for the enclave to use for encryption and signing (see below).
This derivation from the fresh table index $\V{tidx}$ ensures the otypes are fresh and exclusively accessible to the enclave.
A sealing capability $[\perm{SU}, \otype_{a}, \otype_{a}+2, \otype_{a} ]$ for the reserved otypes is written to the first address $b'$ of the data region.

Given an enclave in this layout, \cerisier{} will compute the enclave identity using the hashing primitives described in \Cref{sec:secure-hash-prim}.
Specifically, it will hash the code region of the enclave, excluding the data capability; \ie~the
contents of the memory range $\crange{b+1}{e}$. We additionally include the address $b$ in the
enclave identity.
The data region contents is not included, but the enclave code can later do integrity checks over it if necessary.

After computing the new enclave's identity, \cerisier{} updates its machine state accordingly.
The enclave table is updated at the table index $\V{tidx}$ to contain the new enclave with its identity
and the $\X{EC}$ register is bumped to point to the next free table index.

To allow enclaves to prove their identity to distrusting clients, enclaves receive exclusive access to sealing capabilities for two otypes, after initialization.
Those clients can verify the identity of the enclave using such an otype by querying the machine state using the $\instr{estoreid}$ instruction.
The instruction $\instr{estoreid}\; \V{r_{d}}\; \V{r_{s}}$ reads an otype $\otype$ from the source register $\V{r_{s}}$, queries the enclave table for the identity of the enclave at index $\lfloor \otype/2 \rfloor$, if it exists, and stores the identity in the destination register $\V{r_d}$.
Based on this identity, the client can decide whether it trusts the enclave, typically by comparing it to pre-computed identities of trusted code.

The $\instr{edeinit}\; \V{r_{s}}$ instruction enables early deinitialization using the enclave's
sealing capability. Once deinitialized, an enclave $\V{tidx}$ cannot reappear, as the $\V{EC}$ acts as a bump allocator.
While the original \cheritree{} design includes it, we make no use of it in our case studies.
Nevertheless, with the universal contract (\Cref{sec:ftlr}), we prove that it does not introduce any new attack vectors.

\subsection{SOC Example in Cerisier}
\label{sec:soc_pseudocode}

To demonstrate use of the \cerisier{} attestation primitives, this section revisits the SOC example from \Cref{sec:motiv-exampl-secure}.
\Cref{fig:soc_pseudocode_asm} shows the example in pseudocode \cerisier{} assembly.
The complete listing can be found in the Supplemental Materials, under
\appendixReplace{\Cref{sec:appendix_soc_code}}{B}.

Recall that after the adversary initializes the enclave,
the machine ensures exclusive ownership of the enclave's memory region, installs the (fresh) enclave keys in the data region, creates a new entry in the enclave table, and returns a sentry capability that can be invoked.
When the SOC enclave is invoked, it first fetches the encryption and signing keys from the sealing capability $[\perm{SU}, \otype, \otype + 2, \otype]$ installed in the data section.
Since our enclave only uses signatures, not encryption, we discard $\otype$, and split the sealing capability for $\otype + 1$ into a private signing key (with permission \perm{S}), and a public key (with permission \perm{U}), which will attest that the signed message originates from the enclave.
Finally, the enclave signs the result using the private signing key, and shares it together with its public signing key.

\begin{figure}[t]
  {\small
    \begin{subfigure}[t]{0.50\textwidth}
      \vspace{3pt}
      \begin{equation*}
        \begin{aligned}
          \mathrm{Enclave} \triangleq\enspace
          & keys \gets \mathrm{Data}[0] \\
          & \commentOut{// \; keys := [\perm{SU}, \otype, \otype + 2, \otype]} \\
          & ot \gets \mathrm{getb}(keys) + 1 \\
          & sign\_keys \gets \mathrm{lea}(\mathrm{subseg}(keys, ot, ot + 1), 1) \\
          & \commentOut{// \; sign\_keys := [\perm{SU}, \otype + 1, \otype + 2, \otype+ 1]} \\
          & sign\_pub\_key \gets \mathrm{restrict}(sign\_keys, \perm{U}) \\
          & \commentOut{// \; sign\_pub\_key := [\perm{U}, \otype + 1, \otype + 2, \otype + 1]} \\
          & r \gets 42 \quad \commentOut{ // \mathit{\ computation...}} \\
          & \{r\}_{ot} \gets \mathrm{seal}(r, sign\_keys) \\
          & \textbf{return } (\{r\}_{ot}, sign\_pub\_key)
        \end{aligned}
      \end{equation*}
    \end{subfigure}
    \hfill
    \begin{subfigure}[t]{0.49\textwidth}
    \begin{equation*}
      \begin{aligned}
        ID_{SOC} \triangleq \enspace
        & \mathrm{hash}(\mathrm{Addr}(\mathrm{Enclave}), \mathrm{Enclave}) \\[8pt]
        \mathrm{Client} \triangleq\enspace
        & (\{r\}_{ot}, sign\_pub\_key) \gets \mathrm{Adversary}() \\
        & ot \gets \mathrm{getotype}(\{r\}_{ot}) \\
        & I \gets \mathrm{estoreid}(ot) \\
        & \textbf{if } I \not\eq ID_{SOC} \textbf{ then abort} \\
        & r \gets \mathrm{unseal}(\{r\}_{ot}, sign\_pub\_key) \\
        & \textbf{assert } r = 42 \\[8pt]
        \mathrm{Adversary} \triangleq \enspace
        & \mathrm{enclave^*} \gets \mathrm{einit}(\mathrm{Enclave}, \mathrm{Data}) \\
        & \textbf{return } \mathrm{enclave}()
      \end{aligned}
    \end{equation*}
  \end{subfigure}
  }
\caption{\label{fig:soc_pseudocode_asm}SOC using Cerisier primitives, adapted from \Cref{fig:secure_outsourced_computation_pseudocode}.
The $\instr{lea}$, $\instr{subseg}$, and $\instr{restrict}$ primitives modify the sealing capability's address, bounds, and permissions, respectively.
In the adversary's initialization of the enclave, $\mathrm{Data}$ can be any \perm{RW} memory capability at least one address wide to hold the sealing capability.
Pseudocode adapted from the full listing in the Supplemental Materials, under
\appendixReplace{\Cref{sec:appendix_soc_code}}{B}.
}
\end{figure}
Separately, the client will contain the (pre-computed) static identity of the SOC enclave to check the enclave's attest.
Eventually, the client is resumed with a sealed result and public signing key from the SOC enclave.
The former is really a word sealed with the SOC's signing otype $\mathrm{ot}$, while the latter is an unsealing capability $[\perm{U}, ot, ot + 1, ot]$.
Using the \instr{estoreid} instruction, the client first queries the enclave table to get the dynamic identity corresponding to the otype of the sealed result.
Then, it verifies that the enclave that signed the value is indeed the SOC enclave.
Finally, the sealed value $\{r\}_{ot}$ is unsealed, and verified to be 42.
As the SOC enclave only signs (seals) the value 42, we know that the \instr{assert} will not fail.

\subsection{Attestation Reasoning}
\label{sec:proglog}
Formally capturing the rather special intuitive reasoning model for attestation, as described in \Cref{fig:interactions_scenarios_attest_tc}, is the key challenge for this paper.
Initially, only the client is trusted, that is, we consider a system that is only known to contain the trusted code and unknown, adversarial code (which may or may not contain an enclave).
When the client successfully validates an attested message, the trust model changes.
The validation tells us that an enclave with the specified identity was successfully initialized on the system and has received exclusive authority over the otypes.
Then, we additionally verify the attested enclave, assuming the information captured in the attested identity.
If we can verify that the enclave properly protects the received sealing capabilities and only uses them to sign messages containing correct results, we finally learn something about the received message.

Our approach to formalize this reasoning is to extend the universal contract with an additional guarantee about sealed capabilities produced by adversarial code.
As discussed in \Cref{sec:seal_caps}, the program logic allows associating an otype with a sealing predicate that must be satisfied by all capabilities sealed with the otype.
If a client learns by querying the enclave table that a pair of otypes $(o_1,o_2)$ corresponds to a specific enclave identity $I$, our program logic guarantees that the sealing predicate corresponding to $o_i$ must be a specific pair of predicates $P_{o_1},P_{o_2}$.

For all enclaves of interest, these enclave sealing predicates $P_{o_1},P_{o_2}$ must be provided
to the universal contract in an argument $\customEnclaves$,
mapping enclave identities to the sealing predicates they enforce.
A global \emph{enclave identity invariant} $\systemInvariant(\customEnclaves)$ ties the sealing predicates and the identities of the custom enclaves together,
ensuring that for any enclave in $\customEnclaves$ ever initialized, the sealing predicate will effectively be $P_{o_1},P_{o_2}$.

In addition, each entry in $\customEnclaves$ must be accompanied by program logic proofs that any enclave instance with identity $I$ respects the sealing predicates $P_{o_1}, P_{o_2}$ for the otypes $o_1,o_2$ it receives at initialization.
This requires that the enclave itself uses its sealing capabilities only to seal capabilities satisfying $P_{o_1}$ and $P_{o_2}$, but also that the enclave adequately protects its private state and the sealing capabilities to prevent the adversary from using them to seal other capabilities.

\subsection{Reasoning about Attested Enclaves' Clients}
\label{sec:enclave_res}
To model this intuitive reasoning, our program logic (\Cref{fig:program_logic_syntax}) can represent knowledge about the currently registered enclaves and their identity.
For brevity, we restrict ourselves to the assertions and relations that are necessary to reason
about initially untrusted enclaves communicating with initially trusted clients, even though
\cerisier{} offers additional rules for alternative trust models.
Moreover, we do not discuss enclave deinitialization.

Enclave clients can establish trust in the identity of enclaves using the $\instr{estoreid}$ instruction.
The persistent $\enclaveHist{tidx}{I}$ resource represents the knowledge about an enclave with identity $I$ that has been initialized at table index \V{tidx}.\footnote{The resource $\enclaveHist{tidx}{I}$ does not guarantee that the enclave is still live, i.e.\ it may have already deinitialized itself. The resource continues to be useful because messages from the enclave may still reach communication partners after deinitialization.}

\begin{figure}
  \centering

\begin{mathparpagebreakable}

    \inferH{EStoreIdAttestValid}
    {\X{ValidPC}\lcap{\ppc}{\bpc}{\epc}{\apc}{\vpc} \\
      \X{decode}(n) = \instr{estoreid}\; \I{r_d}\; \I{r_s}\\
      \tidxofot(\otype_a) = \V{tidx} \\
    }
    {
      \stepspecbase
      {
        \begin{array}{l}
          \pc \rmapsto \lcap{\ppc}{\bpc}{\epc}{\apc}{\vpc}
          \ast \ecmapsto{\V{ecn}}
          \\
          {} \ast \apc \amapsto{\vpc} n
          \ast \I{r_s} \rmapsto \otype_a
          \ast \I{r_d} \rmapsto \_ \\
        \end{array}
      }
      {s \ldotp
        \begin{array}{c}
          \left(
          \begin{array}{l}
            \exists I \ldotp
            \pure{s = \Running} \\
            {} \ast \pc \rmapsto \lcap{\ppc}{\bpc}{\epc}{\diff{\apc+1}}{\vpc} \\
            {} \ast \ecmapsto{\V{ecn}} \ast \apc \amapsto{\vpc} n \\
            {  }\ast \I{r_s} \rmapsto \otype_a
            \ast \I{r_d} \rmapsto \diff{I} \\
            \rowstyle{\color{diffColor}}
            {} \ast \enclaveHist{tidx}{I}
            \ast \pure{0 \leq \V{tidx} < \V{ecn}  }
            \rowstyle{\color{black}}
          \end{array}
          \right) \\
          \lor 
          \left(
          \begin{array}{l}
            \pure{s = \Failed}
            \ast \ldots
          \end{array}
          \right)
        \end{array}
      }
    }
\end{mathparpagebreakable}

\caption{\label{fig:wprules_enclaves}
  Rules for the new enclave instructions.
  The ``$\ldots$'' represent precondition resources being returned.
}
\end{figure}

Enclave resources $\enclaveHist{tidx}{I}$ can be obtained by clients using the \ref{EStoreIdAttestValid} rule in \Cref{fig:wprules_enclaves}.
The rule specifies the contract for the enclave attestation instruction \instr{EStoreId}.
The instruction queries the CPU for the identity of the enclave corresponding to a given otype $o_a$.
It has two possible outcomes, depending on whether this otype corresponds to the enclave sealing capabilities of an initialized enclave.
If attestation succeeds, the identity of enclave $I$ was found in the enclave table and is stored in the destination register $\reg{d}$.
A token $\enclaveHist{tidx}{I}$ is created, witnessing that an entry $\V{tidx}$ for the otype $\otype_{a}$, with identity $I$ exists in the enclave table.

To use the knowledge learned using \instr{EStoreId}, enclave clients can combine the $\enclaveHist{tidx}{I}$ resource with user-provided information about a set of enclaves of interest.
This information is represented in the user-defined map of \emph{custom enclaves}
$\customEnclaves : \EId \rightharpoonup \List(\Word) \times Addr \times \iProp \times \iProp$, which maps each known enclave identity $I$ to its
code $\customEnclaves(I).\V{code}$,
the base address $\customEnclaves(I).\V{addr}$ and
the sealing predicates for encryption $\customEnclaves(I).\V{P_{enc}}$ and signing $\customEnclaves(I).\V{P_{sign}}$.
The map is assumed to be well-formed,
\ie{}
the identity should correspond to the hash of the code and base address:
\[\forall I \ldotp \hash(\customEnclaves(I).\V{addr}) \hashconcat
\hash(\customEnclaves(I).\V{code}) = I.\]

\begin{figure}[]
\raggedright

\begin{align*}
  \begin{array}[t]{lcl}
    \systemInvariant(\customEnclaves) & \eqdef & \exists \V{ecn} \ldotp \xxsmallh{A}~ \ecmapsto{\V{ecn}} 
                                          {} \ast \xxsmallh{B}~ \left( \Sep_{\otype \in \crange{0}{2*\V{ecn} + 1}} \exists P \ldotp \sealPred(\otype)(P) \right) \\
   & & {} \ast \xxsmallh{C}~ \customEnclavesPred(\customEnclaves)(\V{ecn})
 \end{array}
  \end{align*}
\begin{align*}
  \begin{array}[t]{l}
   \text{where} \; \customEnclavesPred(\customEnclaves)(\V{ecn}) \eqdef \!\!
   \\ \quad
   \always \forall I \in \dom(\customEnclaves), 0 \leq \V{tidx} < \V{ecn}, \otype \ldotp
    \pure{\tidxofot(\otype) = \V{tidx}} \ast
    \enclaveHist{tidx}{I} \wand
   \\ \quad
     \sealPred(\otype)( \customEnclaves(I).\V{P_{enc}} )
     \ast \sealPred(\otype + 1)( \customEnclaves(I).\V{P_{sign}} )
 \end{array}
\end{align*}
\caption{\label{fig:system_invariant}
  Global invariant $\systemInvariant$ parameterised by user-provided map $\customEnclaves$
  of enclaves of interest (\emph{custom} enclaves).}
\end{figure}

\begin{figure}[]
  \begin{multline*}
    \customEnclavesContract(\customEnclaves) \eqdef
      \forall I \in \V{dom}(\customEnclaves),
      \V{data}, b, e, v, p', b', e', a', v', \otype \ldotp \\
      \left(\begin{aligned}
        &\pure{ b = \customEnclaves(I).\V{addr} } \ast \pure{ \V{tidx} = \tidxofot(\otype)  } \\
        &{} \ast \xxsmallh{1}~ b \amapsto{v} \lcap{p'}{b'}{e'}{a'}{v'}
        \ast \crange{b+1}{e} \amapsto{v} \customEnclaves(I).\V{code} \\
        &\ast \xxsmallh{2}~ b' \amapsto{v'} \left[\perm{SU},\otype, \otype + 2, \otype\right]
        \ast \crange{b'+1}{e'} \amapsto{v'} \V{data} \\
        &\ast \xxsmallh{3}~\sealPred\V{(\otype, \customEnclaves(I).\V{P_{enc}})}
          \ast \xxsmallh{4}~\sealPred\V{(\otype + 1, \customEnclaves(I).\V{P_{sign}})}\\
        &\ast \xxsmallh{5}~\enclaveLive{tidx}{I}
      \end{aligned}  \right)
      \wand \safeV{\lcap{\enter}{b}{e}{a}{v}}
  \end{multline*}

\caption{\label{fig:enclave_contract}
  Contract of custom enclaves, parameterised by user-provided map $\customEnclaves$,
  stating that all enclaves of interest are safe to execute,
  and that they respect the sealing predicate defined in $\customEnclaves$.
}
\end{figure}


Given this user-provided map, the system-wide \emph{enclave identity invariant} $\systemInvariant(\customEnclaves)$ in \Cref{fig:system_invariant} ties the sealing predicates and the identities of the custom enclaves together.
It owns the resource $\ecmapsto{\V{ecn}}$ tracking the number of registered enclaves \xxsmallh{A}, and states that the seal predicates for all otypes up-to the current enclave counter must be initialized \xxsmallh{B}.
Moreover, $\customEnclavesPred$ \xxsmallh{C} states that for all identities in the custom enclaves map $\customEnclaves$, the associated sealing predicates are registered for the enclave's otypes.

Enclave clients that have obtained a $\enclaveHist{tidx}{I}$ resource for the enclave owning a given otype $\otype$, may use the enclave identity invariant to learn that $\sealPred(\otype)( \customEnclaves(I).\V{P_{enc}})$, stipulating that safe capabilities sealed with $\otype$ necessarily satisfy the predicates associated with $I$.
Usually, the sealed capabilities that the client uses have to be safe,
because the adversary acts as intermediary between the enclave and the client.
The case studies in \Cref{sec:case_studies} explain how this is used in practical examples.

\subsection{Reasoning about the Adversary with the Universal Contract}
\label{sec:ftlr}

The $\instr{EStoreId}$ instruction, the enclave identity invariant, and $\enclaveHist{tidx}{I}$ provide enclave clients with extra information they can rely on when receiving sealed values from untrusted code.
However, this information comes at the cost of an additional proof obligation when invoking the adversary.
In particular, we need to prove that the enclaves of interest registered in the custom enclaves map, actually respect the sealing predicates we registered for them.

This condition on registered enclaves is expressed as $\customEnclavesContract$ in \Cref{fig:enclave_contract}.
This contract essentially requires that when an enclave is initialized by adversarial code and its identity matches that of an enclave in the custom enclaves map, then it must respect the sealing predicates registered for it.
More technically, the $\instr{einit}$ instruction does not invoke the initialized enclave itself, but instead creates a sentry capability for it, with exclusive access to the sealing capability for the enclave's otypes.
The custom enclaves contract essentially describes the state of the enclave at that point and requires that it is safe to share with the adversary:
the code region contains the capability pointing to the data region, followed by the instructions themselves \xxsmallh{1},
the data region contains the sealing capability containing the encryption and
signing otypes, followed by the content of the data \xxsmallh{2}.
Importantly, the enclave receives exclusive ownership of this memory in the form of points-to predicates.
The sealing predicates associated with the enclave's keys correspond to those found
in $\customEnclavesContract$ (\xxsmallh{3} - \xxsmallh{4}).
Additionally, the enclave receives the $\enclaveLive{tidx}{I}$ resource representing authority to deinitialize itself \xxsmallh{5}.
Proving this contract amounts to showing that the enclave only ever seals values that satisfy the predicates defined in $\customEnclaves$, and that it does not leak its private keys to untrusted parties.

Apart from this assumption about the registered enclaves, the \cerisier{} universal contract is similar to the \cerise{} one, described in \Cref{thm:univ-contract}.
\begin{theorem}[\cerisier{} Universal Contract]
  \label{thm:univ-contract}
  Let $\customEnclaves$ be a map of custom enclaves,
  such that $\customEnclavesContract(\customEnclaves)$ holds.
  Then,
  $\knowInv{}{\systemInvariant(\customEnclaves)} \proves
      \contspecbase
        {\begin{array}{l}
          \displaystyle\Sep_{
          \substack{r \in \RegName}
          }
          \left(\exists w\ldotp r \rmapsto w \ast \safeV{w}\right)
        \end{array}
        }{}$
\end{theorem}
The contract captures all possible behaviors of arbitrary, adversary code, obeying the guarantees provided by the capability machine, and thus captures capability and attestation safety.
It is a \emph{universal contract}~\cite{huyghebaert_formalizing_2023} since it has no assumptions about the particular code that is executed, i.e., it applies to arbitrary code.
As such, the contract can be used to reason about invocations of untrusted adversarial code.

Notice that proving the contract $\customEnclavesContract$ typically requires recursively invoking the universal contract, for example, to reason about the enclave returning control to its untrusted caller.
Fortunately, this recursion can be resolved by proving the contract using L\"ob induction.
This proof principle, inherited from Iris, allows proving a predicate $P$ by showing that $\later P \proves P$~\cite{iris}, i.e., to prove $P$, one may assume that $P$ holds at later logical steps.

\subsection{Proving the Universal Contract}
\label{sec:ftlr_proof}
The original \cerise{} universal contract is proven by L\"ob induction and case analysis on the first instruction executed by the adversary.
For every instruction in the ISA, one shows that the necessary authority follows from the safety of capabilities that are available.
Furthermore, one proves that the machine state after executing this instruction is also safe, which is done in one of two ways.
In most cases, it suffices to prove that the resulting machine state still contains only safe-to-share values and invoke the L\"ob induction hypothesis.

However, in cases where control is potentially passed to non-adversarial code, e.g., when invoking a sentry capability, we may reach a state where not all capabilities are known to be safe to share.
In such cases, we require another argument to prove the safety of the resulting state.
After invoking a sentry capability, it is converted into an $\perm{RX}$ capability that
may not be safe to share (see \Cref{fig:opsem_excerpt}, in particular, the $\instr{jmp}$ instruction).
However, in this case, we can appeal to the safe-to-execute fact for this $\perm{RX}$ capability that we get from the safety of the original $\perm{E}$ capability (see \Cref{fig:logrel}).

The \cerisier{} proof of the universal contract proceeds along similar lines, adding cases for all the new instructions.
Cases for $\instr{getwtype}$, $\instr{getotype}$, $\instr{estoreid}$, $\instr{hash}$, and $\instr{hashconcat}$ are trivial, since they only produce integers and do not modify system state.
The cases for $\instr{cseal}$ and $\instr{cunseal}$ appeal to the $\safeseal{}$ and $\safeunseal{}$ properties that we learn from the safety of the sealing capabilities (see \Cref{fig:logrel}) to justify that the sealed or unsealed capabilities produced by the instruction remain safe.
Similarly, the case for $\instr{edeinit}$ relies on the authority to deinitialize an enclave that is included in safety of the sealing capability (see \Cref{fig:logrel}).
For brevity, we do not explain the case for $\instr{isunique}$ which involves dealing with the revocation machinery from \Cref{sec:mem_revocation}.

By far the most interesting case is the $\instr{einit}$ instruction, essentially because it creates a new situation where control is passed from the adversary to potentially trusted code.
The operational semantics provides us with a measured identity, a new table entry and a sealing capability for two fresh otypes.
For the fresh otypes, we can register a sealing predicate of choice.
In the case where the measured identity $I$ is not present in the custom enclaves map $\customEnclaves$, we can take $\mathcal{V}$ as sealing predicates and store the enclave ownership predicate in a suitable invariant.
This amounts to treating the new enclave as untrusted, and indeed, we can then prove all produced capabilities safe and appeal to the induction hypothesis.

However, if the measured identity is present in $\customEnclaves$, the invariant $\systemInvariant(\customEnclaves)$ forces us to register the corresponding sealing predicates (rather than $\mathcal{V}$), and the new sealing capability cannot be considered safe-to-share.
It is in this situation a new form of security-boundary-crossing appears, as the untrusted adversarial code metamorphoses into a trusted enclave.
Fortunately, we can make use of $\customEnclavesContract$, which tells us that the freshly defined enclave is safe, at least if we can satisfy the preconditions from \Cref{fig:enclave_contract}.

Many of these preconditions simply describe the resulting machine state after $\instr{einit}$ and are hence easy to satisfy.
However, the points-to predicates for the code and memory region need to be created freshly, intuitively by revoking pre-existing authority over the corresponding memory.
For every address $a$ in these regions, the memory revocation infrastructure from \Cref{sec:mem_revocation} allows us to do this by bumping the current version $v$ of $a$ to $v+1$, allocating $(a,v+1)$ as a fresh logical address and reestablishing correspondence of physical and logical memory by designating $(a,v+1)$ as current.
The memory sweep's positive result provides the necessary evidence that no more references to $(a,v)$ exist in the system.

\subsection{Adequacy}
\label{sec:adequacy}

Our specifications are written and proven in the \cerisier{} program logic.
To draw end-to-end conclusions about executions of our case studies in terms of only the operational
semantics, we prove an \emph{adequacy theorem}.

\begin{theorem}[End-to-end theorem for enclave attestation --- Informal]
  \label{thm:informal_adequacy}
  Suppose a good initial machine state with an initialized client and an (arbitrary) adversary, but no enclave.
  Further suppose that for some custom enclave map $\customEnclaves$,
  the specification $\systemInvariant(\customEnclaves) \proves \contspecbase{P}{}$ holds,
  where $P$ corresponds to the interpretation of the initial state in the program logic.

  Then, the $\instr{assert}$ in the client program will never fail.
\end{theorem}
We provide the formal version of \Cref{thm:informal_adequacy} in the Supplemental Material, under \appendixReplace{\Cref{sec:appendix_adequacy}}{G}.

When proving the end-to-end properties of our case studies \Cref{sec:case_studies},
we always assume a good initial state (which can be obtained, \eg{} after booting the machine),
and the program logic specification is then proved with a $\customEnclaves$ carefully chosen.
When the program interacts with untrusted code, the program logic specification
requires proving that $\customEnclavesContract(\customEnclaves)$ holds
for all registered custom enclaves in $\customEnclaves$.

\section{Case Studies}
\label{sec:case_studies}

In this section, we present case studies of three systems relying on enclaves, implemented using \cerisier{} features and verified using the \cerisier{} program logic.
Two case studies represent trusted execution applications: secure outsourced computation (SOC, \Cref{sec:case_study_soc}) and trusted sensors in a stateful enclave (\Cref{sec:case_study_sensor_readout}).
Both model representative applications of well-known trusted execution applications, studied in the systems literature by, for example, \citet{barbosa_foundations_2016,li_survey_2023} for the former, and \citet{shepherd_establishing_2017,hu_cvshield_2020} for the latter.
Our third case study concerns mutual attestation, where two mutually distrusting enclaves establish trust in each other.
To avoid circularity, such enclaves' identities cannot include the identity that they use to attest each other.
In the case study, we model MAGE: a technique proposed by \citet{MAGE_mutual_attest} to resolve the circularity by attesting enclave code separately from an identity table used to attest other enclaves (\Cref{sec:case_study_mutual_attest}).

\subsection{Verifying the SOC Example}
\label{sec:case_study_soc}
The correctness of the SOC example from \Cref{fig:soc_pseudocode_asm} follows from our adequacy result (\Cref{thm:informal_adequacy})
by instantiating the client's program with the SOC client and proving the Cerisier specification.
The derived property guarantees that the assertions never fail,
we can thus conclude that the unsealed result is 42.
\newcommand{\otres}{\V{ot\_res}}
In a bit more detail, our proof follows the approach sketched in \Cref{sec:soc}.

  First, we instantiate the custom enclave map $\customEnclaves$ to contain a single entry that maps
  the SOC enclave's static identity $\V{ID_{SOC}}$ to the signing predicate $\lambda w \ldotp w = 42$.
  We need to prove that the property associated with the enclave identity holds,
  formally by proving that $\customEnclavesContract(\customEnclaves)$ holds.
  In essence, we must argue that the enclave does not leak its private signing key, and that any
  values signed equal 42.

  Then, we need to prove the specification of the SOC client.
  When the client calls the adversary, we invoke the universal contract \Cref{thm:univ-contract}.
  When the adversary yields to the known client, we expect to receive a sealed result $\sealed{\V{ot}}{\V{r}}$ and the corresponding unsealing key $[ \perm{U}, \V{ot}, \V{ot}+1, \V{ot}]$.
  Moreover, $\safeV{\sealed{\V{ot}}{\V{r}}}$ holds because $\sealed{\V{ot}}{\V{r}}$ was passed by the adversary.
  For the call to \instr{estoreid}, we apply \ref{EStoreIdAttestValid}.
  In the success case, we obtain the enclave resource $\enclaveHist{\V{tidx}}{I}$,
  for the table index $\V{tidx}$ linked to the otype $\V{ot}$, for some identity $I$.
  If $I = \V{ID_{SOC}}$, the system invariant combined with our $\customEnclavesPred$ and $\enclaveHist{\V{tidx}}{I}$ yields $\sealPred(\V{ot})(\lambda w \ldotp w = 42)$.
By combining $\sealPred(\V{ot})(\lambda w \ldotp w = 42)$ and $\safeV{\sealed{\V{ot}}{\V{r}}}$, we know that $\V{r} = 42$.
After unsealing, we can thus safely conclude that the assert holds.

\subsection{Mutual Attestation}
\label{sec:case_study_mutual_attest}
Naively attesting two (or more) enclaves with each other induces a circularity problem since they must each contain a table of static identities for the other enclave(s).
Nevertheless, many applications require a mutual attest to establish secure communication.
\citet{MAGE_mutual_attest} propose a solution to break the circularity.
First, all enclaves involved in the mutual attestation compute a \emph{pre-identity} which hashes only the code, not the identity table.
These pre-identities are gathered in an \emph{identity table}, that is located in memory after the enclaves' code.
When an enclave $E_1$ wants to attest an enclave $E_2$, it can no longer simply look up the expected identity of $E_2$ in its identity table, only the pre-identity.
However, $E_1$ can still reconstruct $E_2$'s identity at runtime by combining $E_2$'s pre-identity with its own copy of the pre-identity table (which should be equal to $E_2$'s copy).
Concretely, it will hash its pre-identity table and combine the result with $E_2$'s pre-identity.
The result is the identity of $E_2$ and can be used to attest $E_2$.
Note that this method requires some algebraic assumptions about \instr{hash},
i.e., hashing the concatenation of the memory contents is equal to concatenating the hash of each, separately:
$\hash(m_{1} \mathrel{++} m_{2}) = \hash(m_{1}) \hashconcat \hash(m_{2})$.

Our mutual attestation case study models this proposed MAGE protocol.
The full code can be found in the Rocq artifact, and simplified pseudocode is provided in the
Supplemental Materials, under \appendixReplace{\Cref{sec:appendix_mutual_attest}}{C}.
Two enclaves A and B, compute 42 and 43, respectively, and attest each other.
When A executes, it signs 42, and yields to the adversary.
It expects the adversary to call B, passing the signed integer.
When B is executed, it expects the adversary to pass the signed value from A.
B then attests A, from which it learns that the signed value must be 42.
Now, B signs 43, and yields to the adversary.
When the adversary then executes A, it expects the signed value from B.
A finally attests B, and learns that the signed value was 43.
Mutual attestation is thus established.
%

Correctness is analogous to the SOC example \Cref{sec:case_study_soc}.
In addition, for the enclaves to mutually attest each other,
$\customEnclavesContract$ holds for enclave A (resp. B)
iff it holds for enclave B (resp. A).
This follows from the L\"ob induction principle.

\subsection{Trusted Sensor Readout}
\label{sec:case_study_sensor_readout}
\label{sec:trust-sensor-read}

The previous two examples only read from memory without writing to it.
We now consider an example that writes to its data section.
The example is an adaptation of the trusted sensor readout example in  CHERI-TrEE~\citep{cheriTree}.
The code can be found in the Supplemental Materials, under \appendixReplace{\Cref{sec:appendix_sensor_readout}}{D}.
Since we do not model I/O, a memory read acts as a stand-in for reading from an MMIO sensor.
The setup is as follows:
the adversary creates a reader enclave that reads from the sensor address, while a separate transformer enclave processes the raw sensor data.
The transformer enclave attests the reader enclave.
Furthermore, a verifier can then attest the transformer enclave to be assured that the values it receives are indeed processed sensor readings.
For both enclaves, we use the initial entry points, created after the adversary executes \instr{einit}, to initialize private state.
We create another sentry to serve as a second entry point.

The initialization of the reader enclave expects a capability $(\perm{RW}, \V{mmio}_b, \V{mmio}_e, \V{mmio}_a)$ to access the sensor.
After ensuring exclusive access with \instr{isunique}, the reader initializes the sensor (here modeled by writing 21 to the sensor's address), and stores the sensor capability in its data section after the sealing keys.
It then creates a sentry capability to a $\V{read}$ entry point in its own code section, signs this sentry, and returns it with the public signing key to the adversary.
The $\V{read}$ entry point fetches the sensor capability from its data, and loads and returns the value at address $\V{mmio}_a$.

The initialization of the transformer enclave expects the sealed sentry to the sensor's $\V{read}$ entry point and the sensor's public signing key which it uses to attest the sensor enclave.
Like the reader enclave, it stores the sentry in its data section and returns a new sentry to a $\V{use}$ entry point.
This second entry point calls the $\V{read}$ entry point of the sensor and processes the value (multiply by 2) before returning.
Finally,
the main program expects a signed sentry to the transformer enclave's $\V{use}$ entry point, which it attests, and asserts that it returns 42.

During the proof, invariants specify the integrity of the data sections.
For instance, after initialization the sensor enclave's data contains a RW-capability pointing to 21.
However, this invariant must be allocated during $\instr{einit}$ before any code runs.
Hence, we make use of a one-shot update to stipulate in a single invariant the shape of the data section before and after calling the enclave.
The second entry point $\V{read}$, has persistent knowledge that the data section has been initialized.

The sensor enclave only ever signs the $\V{read}$ entry point, but in the sealing predicate we specify a triple that must hold for the capability.
Furthermore, we must show that the sealed sentry capability to $\V{read}$ is safe to share with the adversary, i.e., that it satisfies the sealing predicate, or in other words, that it satisfies its triple.
The transformer enclave is verified similarly to the sensor enclave, and the
verification of the main program is also analogous to the SOC example in \Cref{sec:case_study_soc}.

\section{Discussion}

\subsection{Implementing CHERI-TrEE Primitives and Exclusive Memory Ownership}
\label{sec:discussion_memory-sweep}

The CHERI-TrEE authors~\citep{cheriTree} implement the primitives in hardware as an ISA extension to the open-source RISC-V Proteus processor~\citep{bognar23proteus} and describe how the memory sweep integrates into the processor's 5-stage instruction pipeline.
They also implement the primitives in software on the ARM Morello~\citep{morello_2022} capability machine as a trusted hypervisor component providing the primitives as hypercalls.
Both implementations perform an un-optimized eager memory sweep and in this paper, we follow suit to keep verification simple.

Even though this is costly in practice, the CHERI-TrEE authors~\citep{cheriTree} envision it to be acceptable in certain embedded system scenarios.
They also point to work on reducing the cost of the sweep~\citep{joannou2023efficienttaggedmemory} and how it could be run concurrently on multi-core systems~\citep{esswood_cherios_2021}.
The exclusive memory ownership guarantees can also be implemented with other schemes.
For instance, CheriOS~\citep{esswood_cherios_2021} lets a trusted nanokernel manage and hand out \emph{reservations}, from which one can obtain guaranteed exclusively-owned memory.
It uses both a (fast) virtual-memory-remapping and a (slow) GC-like mechanism (similar to CHERIoT's~\citep{cheriot} and Cornucopia~\citep{filardo2020cornucopia,filardo2024cornucopiareloaded}) to reclaim memory.

\subsection{Reasoning about Attestation Generally}
\label{sec:discussion_reasoning-attestation}

Our technical development involves several non-trivial techniques not previously present in Cerise, e.g., for reasoning about sealed capabilities and memory revocation.
Nevertheless, the main innovation is our approach to reasoning about attestation, discussed in Sections~\ref{sec:proglog}-\ref{sec:ftlr}.
The core idea is to treat attestation formally as an additional form of control transfer from untrusted to trusted code.

This core idea is visible in the universal contract (see \Cref{sec:ftlr}), which justifies safety of a control transfer to untrusted code, provided any control transfer from the adversary back to trusted code is safe.
The contract of custom enclaves in \Cref{fig:enclave_contract} models exactly that requirement for the transfer that conceptually happens from untrusted to trusted code when an enclave is freshly attested: safety of the callback produced by $\instr{einit}$ must follow from ownership of the authority the attested enclave obtains (particularly exclusive memory ownership and access to enclave attestation seals).
This is similar to the ``safe to share''-requirement $\safeV{\enter,b,e,a}$ for an enter capability in \Cref{fig:logrel}, since an enter capability lets the adversary transfer control back to trusted code, and this transfer must be proven safe when invoking the universal contract.
\Citet{huyghebaert_formalizing_2023} have shown that the control transfers of ISA security primitives in non-capability-machine settings (particularly RISC-V PMP ecalls and interrupts) can be treated similarly.

Reasoning about other systems offering attestation will require a program logic that can express the guarantees offered by underlying mechanisms (e.g. the lazy GC-like mechanism of \citet{esswood_cherios_2021} or SGX-like lookaside tables \cite{noorman_sancus_2017,brasser_tytan_2015,costan_intel_2016}).
Insofar as those techniques involve a form of revocation of memory ownership, we expect that one should be able to adapt the logical memory technique of \citet{hurSeparationLogicPresence2011}.
After defining a suitable program logic, a universal contract should be defined, proven and applied,
which has assumptions about the control transfer from untrusted code to freshly-attested trusted enclaves,
formalized with a form of contract for custom enclaves.
We believe that our approach generalizes to other attestation systems,
because any attestation system involves a form of control transfer from untrusted code to attested code.

\section{Related Work}

We have discussed the most closely related work along the way.
In this section, we discuss other related work; we restrict ourselves to work that considers verification of low-level software using trusted computing primitives (not just software that implements them)  and verification of capability machine security properties.

Some other work verifies architectural and micro-architectural properties about models of trusted
execution environments (TEE)
\mbox{(micro-)architectures} or their implementation in hardware and firmware.
However, none of these papers can be used to verify end-to-end results about TEE enclaves (or equivalent) and their client(s) in the presence of arbitrary adversary code.
\citet{witharana_formal_2024} use model-checking to verify some security properties about a model of Intel TDX \mbox{(micro-)architectures}.
\citet{nunes_vrased_2019} use a combination of verification techniques to verify hardware and software components of the VRASED TEE system, but \citet{bognar_mind_2022} later discovered various errors in their verification approach.
They did not aim to prove end-to-end results about software relying on attestation.
\citet{ferraiuolo_komodo_2017} use Dafny to prove information flow properties about enclaves in their Komodo system, but as other mentioned work and unlike our work, they do not aim to prove that application software uses the TEE primitives correctly.
Although such work strengthens confidence in correctness of TEE implementations, the authors do not consider how to effectively reason about software making use of TEE primitives.

The \cerise{} capability machine and program logic has been discussed and extended with various features \cite{cerise,van_strydonck_proving_2022,georges_efficient_2021,georges_temps_2022,hammond_morello-cerise_2025} but no attestation primitives have been considered until now.
Other authors have used different approaches for proving weaker properties on larger ISAs \cite{nienhuis_rigorous_2020,bauereiss_verified_2022}.

\section{Conclusion and Future Work}
\label{sec:conclusion-future-work}

We have presented Cerisier, the first program logic for modular reasoning about
trusted, untrusted, and attested code. We have proved Cerisier sound with respect to
a novel operational semantics of a capability machine model with primitives for
attestation and sealing, and we have used Cerisier for defining a universal contract, which
captures both capability safety and local enclave attestation. The universal
contract is formalized using a novel kind of parameterized logical relation. Finally, we have
used Cerisier to prove non-trivial end-to-end properties of representative examples
of trusted computing applications relying on local attestation.
Our work is fully mechanized in \rocq{}, comprising \textasciitilde 35k LoC,
including 2k LoP for the proof of the FTLR case for $\instr{einit}$ sketched in \Cref{sec:ftlr_proof}.
We believe our formal techniques capture intuitive reasoning of trusted execution systems,
and thus that our techniques can inform the design of program logics for other
trusted execution systems.
There are several interesting other lines of future work.

Remote attestation is an important feature of trusted computing that \cerisier{} does not support.
This would require to change our operational semantics and program logic, introducing either a probabilistic \citep[see, e.g., ][]{abate_ssprove_2021} or symbolic \citep[see, e.g.,][]{sumiiLogRelEncryption} model of cryptography.
The resulting changes would be important, but orthogonal to the contribution of this paper.

Similarly to trusted computing, confidential computing derives confidentiality guarantees
about the systems' data and I/O behavior.
We believe that \cerisier{}'s reasoning techniques extend to confidential computing,
and could be based on Cerise's existing binary program logic, which has previously been demonstrated to support confidentiality results.

Finally, it is worth noting that Cerisier currently does not support I/O,
but support for memory-mapped I/O (MMIO) \citep{mmio_cerise} would enable \emph{trusted} and \emph{confidential I/O}, i.e., establishing trust in enclaves performing I/O.
The trusted sensor implementation from Section~\ref{sec:trust-sensor-read} could then be adapted to read from an external sensor rather than from memory.

\begin{acks}
We thank the anonymous PLDI reviewers for their valuable feedback and suggestions.
Denis Carnier holds a PhD fellowship (1S46525N) of the Research Foundation – Flanders (FWO).
This work was supported in part by a Villum Investigator grant (no. 25804 and VIL73403),
Center for Basic Research in Program Verification (CPV), from the VILLUM Foundation.
This research is partially funded by the Internal Funds KU Leuven,
the Cybersecurity Research Program Flanders, and by ERC grant (UniversalContracts, 101040088).
Views and opinions expressed are however those of the author(s) only and do not necessarily reflect
those of the European Union or the European Research Council Executive Agency.
Neither the European Union nor the granting authority can be held responsible for them.

\section*{Data-Availability Statement}
Our results are publicly available~\citep{cerisier-artifact, cerisier-repository}.
  The artifact contains the mechanization in \rocq{} of all definitions, case studies, theorems and proofs discussed in the paper.
  The Supplemental Material is also publicly available~\citep{cerisier-arxiv}.

\end{acks}

\bibliography{biblio.bib}

\ifappendix
\else
\end{document}
\fi

\appendix
\newpage

\section{Operational Semantics}
\label{sec:appendix_opsem}
\begin{figure}[ht!]
  \vspace{-1em}

\small
\bgroup
\setlength\tabcolsep{0.54em}
\noindent\makebox[\textwidth]{\begin{tabular}{|C|L|L|}
  \hline
  \rowstyle{\color{black}}
  $i$ & $\instrsem{i}(\confv)$ & Conditions \\ \hline
  \instr{fail} & $(\K{Failed}, \confv)$ & \\ \hline
  \instr{halt} & $(\K{Halted}, \confv)$ & \\ \hline
  $\instr{mov}\; r\; \rho$
  & $\X{updPC}(\confv[\X{reg}.r \mapsto w])$
  & \!\!\!\begin{tabular}{l}
      $w = \X{getWord}(\confv, \rho)$
    \end{tabular}
  \\ \hline
  $\instr{load}\; r_1\; r_2$ & $\X{updPC}(\confv[\X{reg}.r_1 \mapsto w])$
  & \!\!\!\begin{tabular}{l}
      $\confv.\X{reg}(r_2) = (p, b, e, a)$ and $w = \confv.\X{mem}(a)$ \\
      and $b \le a < e$ and
      $p \in \{ \RO,  \RX, \RW, \RWX \}$
    \end{tabular}
  \\ \hline
  $\instr{store}\; r\; \rho$
  & $\X{updPC}(\confv[\X{mem}.a \mapsto w])$
  & \!\!\!\begin{tabular}{l}
    $\confv.\X{reg}(r) = (p, b, e, a)$ and $b \le a < e$ \\
    and $p \in \{\RW, \RWX\}$
    and $w = \X{getWord}(\confv, \rho)$\\
    \end{tabular}
  \\ \hline
  $\instr{jmp}\; r$
  & {$(\K{Running}, \confv \left[ \X{reg}.\pc \mapsto \V{new} \right])$}
  & $ \V{new} = \X{updPcPerm}(\confv.\X{reg}(r)) $
  \\ \hline
  \multirow{2}{*}{$\instr{jnz}\; r_\mathrm{dst}\; r_\mathrm{cond}$}
      & {$(\K{Running}, \confv \left[ \X{reg}.\pc \mapsto \V{new} \right])$}
      & $\V{new} = \X{updPcPerm}(\confv.\X{reg}(r_\mathrm{dst}))$ and
        $\confv.\X{reg}(r_\mathrm{cond}) \neq 0$
  \\ \cline{2-3}
  & $\X{updPC}(\confv)$
  & $\confv.\X{reg}(r_\mathrm{cond}) = 0$
  \\ \hline
  $\instr{restrict}\; r\; \rho$
  & $\X{updPC}(\confv[\X{reg}.r \mapsto w])$
  & \!\!\!\begin{tabular}{l}
      $\confv.\X{reg}(r) = (p, b, e, a)$ \\ and
      $p' = \X{decodePerm}(\X{getWord}(\confv, \rho))$
      and $p' \permflowsto p$ \\
      and $w = (p', b, e, a)$
    \end{tabular}
  \\ \hline
  \rowstyle{\color{newCerisierColor}}
  $\instr{restrict}\; r\; \rho$
  & $\X{updPC}(\confv[\X{reg}.r \mapsto w])$
  & \!\!\!\begin{tabular}{l}
    $\confv.\X{reg}(r) = [\V{sp},\otype_{b},\otype_{e},\otype_{a}]$ \\
    and $\V{sp'} = \X{decodeSealPerm}(\X{getWord}(\confv, \rho))$ \\
    and $\V{sp'} \permflowsto \V{sp}$
    and $w = [\V{sp'},\otype_{b},\otype_{e},\otype_{a}]$
    \end{tabular}
  \\ \hline
  \rowstyle{\color{black}}
  $\instr{subseg}\; r\; \rho_1 \; \rho_2$
  & $\X{updPC}(\confv[\X{reg}.r \mapsto w])$
  & \!\!\!\begin{tabular}{l}
      $\confv.\X{reg}(r) = (p, b, e, a)$ and for $i \in \{1,2\}$, \\
    $z_i = \X{getWord}(\confv, \rho_i)$ and $z_i \in \ZZ$ \\
    and $b \le z_1 < \X{AddrMax}$ and $0 \le z_2 \le e$ \\
    and $p \neq \enter$ and $w = (p, z_1, z_2, a)$
    \end{tabular}
  \\ \hline
  \rowstyle{\color{newCerisierColor}}
  $\instr{subseg}\; r\; \rho_1 \; \rho_2$
  & $\X{updPC}(\confv[\X{reg}.r \mapsto w])$
  & \!\!\!\begin{tabular}{l}
      $\confv.\X{reg}(r) = [\V{sp},\otype_{b},\otype_{e},\otype_{a}]$ and for $i \in \{1,2\}$, \\
    $z_i = \X{getWord}(\confv, \rho_i)$ and $z_i \in \ZZ$ \\
    and $\otype_{b} \le z_1 < \X{OTypeMax}$ and $0 \le z_2 \le \otype_{e}$ \\
    and $w = [\V{sp},z_{1},z_{2},\otype_{a}]$
    \end{tabular}
  \\ \hline
  \rowstyle{\color{black}}
  $\instr{lea}\; r \; \rho$
  & $\X{updPC}(\confv[\X{reg}.r \mapsto w])$
  & \!\!\!\begin{tabular}{l}
      $\confv.\X{reg}(r) = (p, b, e, a)$ and $z = \X{getWord}(\confv, \rho)$ \\
      and $p \neq \enter$
      and $w = (p, b, e, a + z)$
          \end{tabular}
  \\ \hline
  \rowstyle{\color{newCerisierColor}}
  $\instr{lea}\; r \; \rho$
  & $\X{updPC}(\confv[\X{reg}.r \mapsto w])$
  & \!\!\!\begin{tabular}{l}
      $\confv.\X{reg}(r) = [\V{sp},\otype_{b},\otype_{e},\otype_{a}]$ and $z = \X{getWord}(\confv, \rho)$ \\
      and $w = [\V{sp},\otype_{b},\otype_{e},\otype_{a} + z]$
          \end{tabular}
  \\ \hline
  \rowstyle{\color{black}}
  $\instr{add}\; r \; \rho_1 \; \rho_2$
  & $\X{updPC}(\confv[\X{reg}.r \mapsto z])$
  & \!\!\!\begin{tabular}{l}
      for $i \in \{1, 2\}$, $z_i = \X{getWord}(\confv, \rho_i)$ \\
      and $z_i \in \ZZ$ and $z = z_1 + z_2$
    \end{tabular}
  \\ \hline
  $\instr{sub}\; r \; \rho_1 \; \rho_2$
  & $\X{updPC}(\confv[\X{reg}.r \mapsto z])$
  & \!\!\!\begin{tabular}{l}
      for $i \in \{1, 2\}$, $z_i = \X{getWord}(\confv, \rho_i)$ \\
      and $z_i \in \ZZ$ and $z = z_1 - z_2$
    \end{tabular}
  \\ \hline
  $\instr{lt}\; r \; \rho_1 \; \rho_2$
  & $\X{updPC}(\confv[\X{reg}.r \mapsto z])$
  & \!\!\!\begin{tabular}{l}
      for $i \in \{1, 2\}$, $z_i = \X{getWord}(\confv, \rho_i)$ \\
      and $z_i \in \ZZ$ and if $z_1 < z_2$ then $z = 1$ else $z = 0$
    \end{tabular}
  \\ \hline
  $\instr{getp}\; r_1\; r_2$
  & $\X{updPC}(\confv[\X{reg}.r_1 \mapsto z])$
                               & \!\!\!\begin{tabular}{l}
    $\left(
    \begin{tabular}{c}
      $\confv.\X{reg}(r_2) = (p, \_, \_, \_)$
      and $z = \X{encodePerm}(p)$
    \end{tabular}
    \right)
    $\\
    \newCerisier{
    or
    $\left(
    \begin{tabular}{c}
    $\confv.\X{reg}(r_2) = [\V{sp}, \_, \_, \_]$ \\
    and $z = \X{encodeSealPerm}(\V{sp})$
    \end{tabular}
    \right)$
    }
    \end{tabular}
  \\ \hline
  $\instr{getb}\; r_1\; r_2$
  & $\X{updPC}(\confv[\X{reg}.r_1 \mapsto b])$
  & \!\!\!\begin{tabular}{l}
      $\confv.\X{reg}(r_2) = (\_, b, \_, \_)$
    \newCerisier{or $\confv.\X{reg}(r_2) = [\_, b, \_, \_]$}
    \end{tabular}
  \\ \hline
  $\instr{gete}\; r_1\; r_2$
  & $\X{updPC}(\confv[\X{reg}.r_1 \mapsto e])$
  & \!\!\!\begin{tabular}{l}
      $\confv.\X{reg}(r_2) = (\_, \_, e, \_)$
    \newCerisier{or $\confv.\X{reg}(r_2) = [\_, \_, e, \_]$}
    \end{tabular}
  \\ \hline
  $\instr{geta}\; r_1\; r_2$
  & $\X{updPC}(\confv[\X{reg}.r_1 \mapsto a])$
  & \!\!\!\begin{tabular}{l}
      $\confv.\X{reg}(r_2) = (\_, \_, \_, a)$
    \newCerisier{or $\confv.\X{reg}(r_2) = [\_, \_, \_, a]$}
    \end{tabular}
  \\ \hline
  \rowstyle{\color{newCerisierColor}}
  $\instr{getotype}\; r_1\; r_2$
  & $\X{updPC}(\confv[\X{reg}.r_1 \mapsto z])$
  & \!\!\!\begin{tabular}{l}
      if $\confv.\X{reg}(r_2) = \sealed{\otype}{w}$
      then $z = \otype$ else $z = -1$
    \end{tabular}
  \\ \hline
  $\instr{getwtype}\; r_1\; r_2$
  & $\X{updPC}(\confv[\X{reg}.r_1 \mapsto z])$
  & \!\!\!\begin{tabular}{l}
    $\confv.\X{reg}(r_2) = w$ and
    $z = \X{encodeWordType}(w)$
    \\
    \end{tabular}
  \\ \hline
  sealing instructions & \ldots & \ldots \\ \hline
  enclave instructions & \ldots & \ldots \\ \hline
  %
  \rowstyle{\color{black}}
  \_ & $(\K{Failed}, \confv)$ & otherwise
  \\ \hline
\end{tabular}}
\egroup

\caption{Operational semantics: execution of a single instruction.}
\end{figure}

\begin{figure}[ht!]
\small
\bgroup
\setlength\tabcolsep{0.54em}
\noindent\makebox[\textwidth]{\begin{tabular}{|C|L|L|}
  \hline
  $i$ & $\instrsem{i}(\confv)$ & Conditions \\ \hline
  \rowstyle{\color{newCerisierColor}}
  %
  $\instr{cseal}\; r_{d} \; r_{1} \; r_{2}$
  & $\X{updPC}(\confv[\X{reg}.r_{d} \mapsto \sealed{\otype_{a}}{\V{sc}}]$
  & \!\!\!\begin{tabular}{l}
    $\confv.\X{reg}(r_1) = [\V{sp},\otype_{b},\otype_{e},\otype_{a}]$
    and $\confv.\X{reg}(r_2) = \V{sc}$ \\
    and $\V{sp} \in \{\perm{S}, \perm{SU}\}$
    and $\otype_{b} \le \otype_{a} < \otype_{e}$ \\
    \end{tabular}
  \\ \hline
  $\instr{cunseal}\; r_{d} \; r_{1} \; r_{2}$
  & $\X{updPC}(\confv[\X{reg}.r_{d} \mapsto \V{sc}]$
  & \!\!\!\begin{tabular}{l}
    $\confv.\X{reg}(r_1) = [\V{sp},\otype_{b},\otype_{e},\otype_{a}]$
    and $\confv.\X{reg}(r_2) = \sealed{\otype_{a}}{\V{sc}}$ \\
    and $\V{sp} \in \{\perm{U}, \perm{SU}\}$
    and $\otype_{b} \le \otype_{a} < \otype_{e}$ \\
    \end{tabular}
  \\ \hline

  %
  $\instr{einit}\; r_{1}\; r_{2}$
  & \!\!\!\begin{tabular}{l}
    {$\X{updPC}(\confv[$} \\
    \quad{$\X{mem}.b \mapsto (\RW,b',e',a')$}, \\
    \quad{$\X{mem}.b' \mapsto [\perm{su},\otype_{a},\otype_{a+2},\otype_{a}]$}, \\
    \quad{$\X{etbl}.\V{tidx} \mapsto I$}, \\
    \quad{$\X{EC} \mapsto \confv.\X{EC} + 1$}, \\
    \quad{$\X{reg}.r_{1} \mapsto (\enter, b, e, b+1)$}, \\
    \quad{$\X{reg}.r_{2} \mapsto 0$])} \\
  \end{tabular}

  & \!\!\!\begin{tabular}{l}
    $r_{1} \neq \pc $
    and $\confv.\X{reg}(r_{1}) = (\RX,b,e,a)$
    and $b < e$ \\
    and $\confv.\X{reg}(r_{2}) = (\RW,b',e',a')$
    and $b' < e'$ \\
    and $\sweepr(\confv)(r_{1})$
    and $\sweepr(\confv)(r_{2})$ \\
    and $\left(\forall l \in \crange{b+1}{e}.\;
    \confv.\X{mem}(l) \in \ZZ\right)$\\
    and $I = \hash(b) \hashconcat \hash(\confv.\X{mem}(\; \crange{b+1}{e} \; ))$ \\
    and $\V{tidx} = \freshtidx(\confv)$
    and $\otype_{a} = \confv.\X{EC}*2$\\
    \end{tabular}
  \\ \hline
  $\instr{edeinit}\; r$
  & \!\!\!\begin{tabular}{l}$\X{updPC}(\confv[\X{etbl}.\X{tidx} \mapsto \emptyset ])$ \end{tabular}
  & \!\!\!\begin{tabular}{l}
    $\confv.\X{reg}(r) = [\perm{su},\otype_{a},\otype_{a+2},\_]$ \\
    and $\X{tidx} = \tidxofot(\otype_{a})$
    and $\confv.\X{etbl}(tidx) = I$\\
    \end{tabular}
  \\ \hline
  $\instr{estoreid}\; r_d \; r_s$
  & \!\!\!\begin{tabular}{l}$\X{updPC}(\confv[\X{reg}.r_{d} \mapsto I ])$ \end{tabular}

  & \!\!\!\begin{tabular}{l}
    $\confv.\X{reg}(r_{s}) = \otype_{s}$
    and $\X{tidx} = \tidxofot(\otype_{s})$\\
    and $\confv.\X{etbl}(tidx) = I$\\
    \end{tabular}
  \\ \hline
  $\instr{isunique}\; r_{d}\; r_{s}$
  & \!\!\!\begin{tabular}{l}$\X{updPC}(\confv[\X{reg}.r_{d} \mapsto z])$\end{tabular}
  & \!\!\!\begin{tabular}{l}
    $\left(
    \begin{tabular}{c}
      $\confv.\X{reg}(r_{s}) = (p,b,e,a)$ \\
      $\lor$
      $\confv.\X{reg}(r_{s}) = \{(p,b,e,a)\}_{\otype_{i}}$
    \end{tabular}
    \right)
    $\\
    and
    $\left(
    \begin{tabular}{ll}
      if & $\sweepr(\confv)(r_{s})$\\
      then & $z=1$ \quad else $z=0$
    \end{tabular}
    \right)$
  \end{tabular}
  \\ \hline
  $\instr{hash}\; r_{d}\; r_{s}$
  & \!\!\!\begin{tabular}{l}$\X{updPC}(\confv[\X{reg}.r_{d} \mapsto \hash(w)])$\end{tabular}
  & \!\!\!\begin{tabular}{l} $\confv.\X{reg}(r_{s}) = w$ \end{tabular}
  \\ \hline
  $\instr{hashconcat}\; r_{d} \; \rho_1 \; \rho_2$
  & $\X{updPC}(\confv[\X{reg}.r_{d} \mapsto z])$
  & \!\!\!\begin{tabular}{l}
      for $i \in \{1, 2\}$, $z_i = \X{getWord}(\confv, \rho_i)$ \\
      and $z_i \in \ZZ$ and $z = z_1 \hashconcat z_2$
    \end{tabular}
  \\ \hline
  \rowstyle{\color{black}}
  \ldots & \ldots & \ldots
  \\ \hline
\end{tabular}}
\egroup
\caption{Execution of a single instruction for sealing, unsealing and enclaves.}
\end{figure}

\begin{figure}[ht!]
  \small
  \begin{mathpar}
  \inferrule[ExecSingle]{}
  {{\begin{array}{l}
    (\K{Running}, \confv) \rightarrow \left\{
    \begin{array}{l}
      \instrsem{\mathit{decode}(z)}(\confv)
      \quad \begin{array}[t]{ll}
        \mathrm{if}  \!\!\!\!&
        \confv.\mathrm{reg(\pc)} = (p, b, e, a) \wedge b \le a < e \wedge {} \\
        & p \in \{\RX, \RWX \} \wedge \confv.\mathrm{mem(\V{a})} = z \\[0.5em]
      \end{array}
      \\
      (\K{Failed},\, \confv) \qquad\quad\;\, \mathrm{otherwise}
    \end{array}
    \right.
    \end{array}}
}

  \X{updPC}(\confv) = \left\{
    \begin{array}{ll}
      (\K{Running}, \confv[\X{reg}.\pc \mapsto (p, b, e, a+1)])
      & \text{if } \confv.\X{reg}(\pc) = (p, b, e, a)\\ 
      (\K{Failed}, \confv) & \text{otherwise}
    \end{array}
    \right.

  \X{getWord}(\confv, \rho) = \left\{
    \begin{array}{ll}
      \rho & \text{if } \rho \in \ZZ \\
      \confv.\X{reg}(\rho) & \text{if } \rho \in \X{RegName}
    \end{array}
  \right.

  \X{updPcPerm}(w) = \left\{
    \begin{array}{ll}
      (\RX, b, e, a) & \text{if } w = (\enter, b, e, a) \\
      w & \text{otherwise}
    \end{array}
  \right.

  \color{newCerisierColor}
  \overlap(w_1)(w_2) =
  \left\{
  \begin{array}{ll}
    \crange{b_1}{e_1} \cap  \crange{b_2}{e_2}
    &\begin{array}{l}
      \text{if }
      \left(w_1 = (p_1,b_1,e_1,a_1) \lor w_1 = \sealed{\otype}{(p_1,b_1,e_1,a_1)}\right)
      \\ \text{and}
      \left(w_2 = (p_2,b_2,e_2,a_2) \lor w_2 = \sealed{\otype}{(p_2,b_2,e_2,a_2)}\right)
    \end{array}
    \\ \bottom & \text{otherwise}
  \end{array}
  \right.

  \sweepr(\confv)(r_s) =
  \begin{array}{l}
    \forall r \in \V{dom}(\confv.\X{reg} \setminus r_s).\;
    \neg \overlap(\confv.\X{reg}(r), \confv.\X{reg}(r_s))\\
     {} \land \forall a \in \V{dom}(\confv.\X{mem}).\;
    \neg \overlap(\confv.\X{mem}(a), \confv.\X{reg}(r_s))
  \end{array}

  \tidxofot(\otype_a) = \left\{
    \begin{array}{ll}
      \otype_a \slash 2 & \text{if $\V{is\_even}(\otype_a)$}\\
      (\otype_a-1) \slash 2 & \text{otherwise}
    \end{array}
  \right.

  \freshtidx(\confv) = \confv.\X{EC}
\color{black}
  \end{mathpar}

\caption{Operational semantics helper.}
\end{figure}

\begin{figure}[ht!]

  \begin{subfigure}[t]{0.45\textwidth}
    \centering
    \begin{tikzcd}[tips=false,column sep=1em,row sep=1em]
      & \perm{RWX} \ar{dl} \ar{dr} & & \\
      \perm{RW} \ar{dr} & & \perm{RX} \ar{dl} \ar{d} & \\
      & \perm{RO} \ar{d} & \textsc{e} \ar{dl} & \\
      & \perm{o} & & \\
    \end{tikzcd}
    \caption{Lattice defining the $\permflowsto$ relation for memory capability permissions.}
  \end{subfigure}
  \hfill
  \begin{subfigure}[t]{0.45\textwidth}
    \centering
    \begin{tikzcd}[tips=false,column sep=1em,row sep=1em]
      & \perm{SU} \ar[dr] \ar[dl] &  \\
      \perm{S} \ar[dr]  &  & \perm{U} \ar[dl] \\
      & \perm{o}  &
    \end{tikzcd}
    \caption{Lattice defining the $\permflowsto$ relation for sealing capability permissions.}
  \end{subfigure}

  \caption{
    Lattice of permission.
    We have $p_{1} \permflowsto p_{2}$ if there is a path going up from $p_{1}$ to $p_{2}$ in
    the diagram.
  }

\end{figure}


\clearpage

\section{SOC Code}
\label{sec:appendix_soc_code}
\begin{figure}[!ht]
  \begin{capasmsmall}
; Expected initial register file:
; PC := (RWX, sensor, sensor_end, sensor+1)
; r0 := (RWX, adv, adv_end, adv)
main:
   ; create callback sentry
   mov r1 pc
   lea r1 (callback - main)
   restrict r1 E
   ; jump to adversary
   jmp r0
callback:
   ; PC := (RX, main, main_end, callback)
   ; r0 := <expects sealed value>
   ; r3 contains the failing capability, if any check fails
   mov r3 pc
   mov r4 r3
   lea r3 (fails - callback)
   ; check that r0 contains a capability
   getotype r2 r0
   sub r2 r2 (-1)
   mov r5 pc
   lea r5 4
   jnz r5 r2
   jmp r3
   ; attestation
   getotype r2 r0
   estoreid r4 r2
   ; check otype(w_res) against identity of the enclave
   sub r4 r4 hash_enclave
   jnz r3 r4
   ; get returned value and assert it to be 42
   unseal r1 r1 r0
   mov r0 r5
   geta r4 r1
   mov r5 42
   mov r1 r3
   lea r1 1
   load r1 r1
   ASSERT r1 r4 r5
   halt
fails:
   fail
data:
   (RO, b_lt, e_lt, b_lt); linking table_cap
   (RWX, main, main_end, data) ; writable cap of the main program
main_end:
\end{capasmsmall}
\caption{Full Cerisier code of the SOC client}
\label{fig:code_trusted_compute}
\end{figure}


\begin{figure}[!ht]
  \begin{capasmsmall}
; Expected initial register file:
; PC := (RWX, enclave, enclave_end, enclave + 1)
enclave:
   (RW, data_enclave, data_enclave_end, data_enclave)
   ; get signing sealing key
   mov r1 pc
   lea r1 (-1)
   load r1 r1
   getb r2 r1
   geta r3 r1
   sub r2 r2 r3
   lea r1 r2
   load r1 r1
   gete r3 r1
   sub r2 r3 1
   subseg r1 r2 r3
   ; store the result (42) in a o-permission capability and sign it
   mov r2 pc
   geta r3 r2
   sub r3 42 r3
   lea r2 r3
   restrict r2 O
   lea r1 1
   seal r2 r1 r2
   ; share the signed value and the unsealing key to the adversary
   restrict r1 (Unseal)
   jmp r0
enclave_end:
data_enclave:
   (SU, ot_enc_tc, ot_enc_tc+2, ot_enc_tc)
data_enclave_end:
\end{capasmsmall}
\caption{Full Cerisier code of the SOC enclave}
\label{fig:code_enclave_trusted_compute}
\end{figure}


\begin{figure}[!ht]
  \begin{capasmsmall}
    ; PC contains ()
    ; r1 contains the callback to main
    ; everything else contains garbage values
adv_start:
    ; create the enclave capability
    mov r2 pc
    getb r3 pc
    add r3 r3 (enclave-adv_start)
    getb r4 pc
    add r4 r4 (enclave_end-adv_start)
    subseg r2 r3 r4
    lea r2 (enclave-adv_start+1)
    mov r5 r2
    restrict r2 RX
    mov r0 0
    ; restrict adv such that it doesn't intersect with the enclave
    getb r4 pc
    subseg pc r4 r3
    ; get the data capability
    lea r5 -1
    load r4 r5
    store r5 0
    mov r5 0
    ; initialise the enclave
    einit r2 r4
    ; store r1 in r31
    mov r31 r1
    ; prepare callback, and calls it to get the value
    mov r0 pc
    lea r0 4
    restrict r0 (E, Global)
    jmp r2
    ; callback
    ; < r1 contains unsealing cap >
    ; < r2 contains sealed cap >
    mov r0 r2
    jmp r31
enclave:
   (RW, data_enclave, data_enclave_end, data_enclave)
....
enclave_end:
data_enclave:
....
data_enclave_end:
adv_end:
\end{capasmsmall}
\caption{Full Cerisier code of a possible adversary for the SOC example}
\label{fig:code_adversary_trusted_compute}
\end{figure}


\clearpage

\section{Mutual Attestation Code}
\label{sec:appendix_mutual_attest}
\begin{figure}[!ht]
  \begin{capasmsmall}
; Expected initial register file:
; PC := (RWX, main, main_end, main)
; r0 := < sealed cap of enclave A >
; r1 := < unsealing cap for A >
; r2 := < sealed cap of enclave B >
; r3 := < unsealing cap for B >
main:
    ; --- mutual_attestation_main_attest_or_fail r0 #A ---
    ; wsealed_A := get_reg(r0)
    ; otype_wsealed_A := get_otype(wsealed_A)
    ; id_A := attest(otype_wsealed_A)
    ; if ( id_A != #A ) then fail else continue

    ; --- mutual_attestation_main_get_confirm_or_fail r0 r1 ---
    ; unseal_A := get_reg(r1)
    ; unsealed_A := unseal(wsealed_A, unseal_A)
    ; if (get_base(unsealed_A) 

    ; --- mutual_attestation_main_attest_or_fail r2 #B ---
    ; wsealed_B := get_reg(r2)
    ; otype_wsealed_B := get_otype(wsealed_B)
    ; id_B := attest(otype_wsealed_B)
    ; if ( id_B != #B ) then fail else continue

    ; --- mutual_attestation_main_get_confirm_or_fail r2 r3 ---
    ; unseal_B := get_reg(r2)
    ; unsealed_B := unseal(wsealed_B, unseal_B)
    ; if (get_base(unsealed_B) 

    ; assert (get_addr(unsealed_A) == 1)
    ; assert (get_addr(unsealed_B) == 1)
    ; halt
main_end:
\end{capasmsmall}
\caption{Cerisier pseudocode of the Mutual Attestation client}
\label{fig:code_mutual_attest_main}
\end{figure}


\begin{figure}[!ht]
  \begin{capasmsmall}
; Expected initial register file:
; PC := (RWX, enclave_A, enclave_A_end, enclave_A)
; r0 := < return pointer >
enclave_A:
    ; --- create {(O,a,a+1,42)}_signed_A ; with a
    ; --- create return to callback_enclave_A ---
    ; --- return(r0, {(O,a,a+1,42)}_signed_A) ---

callback_enclave_A:
; PC := (RWX, enclave_A, enclave_A_end, callback_enclave_A)
; r0 := < return pointer >
; r1 := < sealed cap of enclave B >
; r2 := < unsealed cap for B >
    ; idT_B := data_A(B)
    ; idT := hash(data_A::data_A_end)
    ; expected_id_B := ( idT_B|| idT)

    ; wsealed_B := get_reg(r1)
    ; otype_wsealed_B := get_otype(wsealed_B)
    ; id_B := attest(otype_wsealed_B)
    ; if ( id_B != expected_id_B ) then fail else continue

    ; unseal_B := get_reg(r2)
    ; unsealed_B := unseal(wsealed_B, unseal_B)
    ; if (get_base(unsealed_B) \% 2 != 0) then fail else continue
    ; if (get_addr(unsealed_B) != 43) then fail else continue

    ; --- create {(O,a+1,a+2,1)}_signed_A ; with a
    ; --- return(r0, {(O,a+1,a+2,1)}_signed_A, pub_sign_A) ---
enclave_A_end:

data_A:
    #[enclave_A::enclave_A_end]
    #[enclave_B::enclave_B_end]
data_A_end:
\end{capasmsmall}
\caption{Cerisier pseudocode of the Mutual Attestation Enclave A}
\label{fig:code_mutual_attest_A}
\end{figure}

\begin{figure}[!ht]
  \begin{capasmsmall}
; Expected initial register file:
; PC := (RWX, enclave_B, enclave_B_end, enclave_B)
; r0 := < return pointer >
; r1 := < sealed cap of enclave B >
; r2 := < unsealed cap for B >
enclave_B:
    ; idT_A := data_B(A)
    ; idT := hash(data_B::data_B_end)
    ; expected_id_A := ( idT_A || idT)

    ; wsealed_A := get_reg(r1)
    ; otype_wsealed_A := get_otype(wsealed_A)
    ; id_A := attest(otype_wsealed_A)
    ; if ( id_A != expected_id_A ) then fail else continue

    ; unseal_A := get_reg(r2)
    ; unsealed_A := unseal(wsealed_A, unseal_A)
    ; if (get_base(unsealed_A) \% 2 != 0) then fail else continue
    ; if (get_addr(unsealed_A) != 42) then fail else continue

    ; --- create {(O,a,a+1,43)}_signed_B ; with a\%2=0 ---
    ; --- create {(O,a+1,a+2,1)}_signed_B ; with a\%2=0 ---
    ; --- return(r0, {(O,a,a+1,42)}_signed_B, {(O,a+1,a+2,1)}_signed_B, pub_sign_B) ---
enclave_B_end:

data_B:
    #[enclave_A::enclave_A_end]
    #[enclave_B::enclave_B_end]
data_B_end:
\end{capasmsmall}
\caption{Cerisier pseudocode of the Mutual Attestation Enclave B}
\label{fig:code_mutual_attest_B}
\end{figure}


\clearpage

\section{Sensor Readout Code}
\label{sec:appendix_sensor_readout}
\begin{figure}[!ht]
\begin{capasmsmall}
; Expected initial register file:
; PC := (RWX, main, main_end, main)
; r0 := (RWX, adv, adv_end, adv)
main:
   mov r1 pc
   lea r1 (callback - main)
   restrict r1 E
   ; jump to adversary
   jmp r0
callback:
   ; PC := (RX, main, main_end, callback)
   ; r0 := <expects sealed sentry to client use entry point>
   ; r1 := <expects client enclave's public signing key>
   mov r5 pc
   ; r5 contains failing capability, if any check fails
   lea r5 (fails - callback)
   ; check that r0 contains a sealed capability
   getotype r2 r0
   sub r3 t2 (-1)
   mov r4 pc
   lea r4 4
   jnz r4 r3
   jmp r5
   ; attestation
   estoreid r3 r2 ; r2 still contains the otype
   ; check the identity of the client enclave
   sub r3 r3 hash_client_enclave
   jnz r5 r3
   ; jump to client enclave
   unseal r1 r1 r0
   mov r0 pc
   lea r0 3 ; r0 = return_point
   jmp r1
return_point:
   mov r1 r5
   mov r4 r2
   mov r5 42
   lea r1 1
   load r1 r1
   ASSERT r1 r4 r5
   halt
fails:
   fail
data:
   (RO, b_lt, e_lt, b_lt); linking table_cap
   (RWX, main, main_end, data) ; writable cap of the main program
main_end:
\end{capasmsmall}
\caption{Full Cerisier code of the trusted sensor readout client}
\label{fig:code_trusted_sensor_readout_main}
\end{figure}


\begin{figure}[!ht]
\begin{capasmsmall}
; Expected initial register file:
; PC := (RWX, sensor, sensor_end, sensor + 1)
; r1 := <expects rw capability to sensor's mmio address>
sensor:
   (RW, data_sensor, data_sensor_end, data_sensor)
   mov r2 pc
   lea r2 (fails - (sensor + 1))
   getwtype r3 r1
   sub r3 r3 wcaptype
   jnz r2 r3     ; jump to fails if r1 is not a capabilitu
   getp r3 r1
   sub r3 r3 RW
   jnz r2 r3     ; jump to fails if r1 does not have RW permissions
   isunique r3 r1
   sub r3 1 r3
   jnz r2 r3     ; jump to fails if not unique
   store r1 21   ; initialize sensor
   mov r3 r2
   lea r3 (sensor - fails)
   load r3 r3
   getb r4 r3
   geta r5 r3
   sub r4 r4 r5
   jnz r2 r4     ; jump to fails if cursor of data capability is not data_sensor
   lea r3 1
   store r3 r1   ; store mmio capability
   lea r3 (-1)
   load r1 r3
   lea r1 1
   lea r2 (read - fails)
   restrict r2 e ; construct sentry capability for read
   seal r2 r1 r2 ; sign the sentry capability
   geta r3 r1
   gete r4 r1
   subseg r1 r3 r4
   restrict r1 u ; signing public key
   jmp r0
read:
   mov r1 PC
   lea r1 (sensor - read)
   load r1 r1   ; load data capability
   lea r1 1
   load r1 r1   ; load sensor capability
   load r2 r1   ; r_t2 = sensor value (21)
   geta r1 r1   ; r_t3 = mmio_a <mmio address>
   jmp r0
fails:
   fail
sensor_end:
data_sensor:
  (SU, ot_sens_enc, ot_sens_enc+2, ot_sens_enc)
  (RW, mmio_b, mmio_e, mmio_a)
data_sensor_end:
\end{capasmsmall}
\caption{Full Cerisier code of the sensor enclave}
\label{fig:code_trusted_sensor_readout_sensor}
\end{figure}


\begin{figure}[!ht]
\begin{capasmsmall}
; Expected initial register file:
; PC := (RWX, transformer, transformer_end, transformer + 1)
; r1 := <expects sensor enclave's public signing key>
; r2 := <expects sealed sentry to sensor enclave's read entry point>
transformer:  (RW, data_transformer, data_transformer_end, data_transformer)
              mov r3 pc
              lea r3 (fails - (transformer + 1))
              unseal r2 r1 r2 ; unseal read entry point
              geta r4 r1
              estoreid r1 r4  ; attest sensor enclave
              sub r1 r1 hash_sensor_enclave
              jnz r3 r1
              lea r3 (transformer - fails)
              load r1 r3      ; load data capability
              getb r4 r1
              geta r5 r1
              sub r4 r4 r5
              lea r1 r4
              lea r1 1
              store r1 r2     ; store read sentry in private data
              lea r1 (-1)
              load r2 r1      ; load sealing keys
              lea r2 1
              lea r3 (use - transformer)
              restrict r3 E   ; create sentry to use
              seal r1 r2 r3   ; sign the sentry
              geta r3 r2
              gete r4 r2
              subseg r2 r3 r4
              restrict r2 U   ; public signing key
              jmp r0
use:          mov r1 pc
              lea r1 (transformer - use)
              load r1 r1      ; load data capability
              getb r2 r1
              geta r3 r1
              sub r2 r2 r3
              lea r1 r2
              lea r1 1
              load r1 r1      ; load read sentry
              mov r3 r0       ; save return pointer in register unclobbered by read
              mov r0 pc
              lea r0 3        ; r0 = return_point
              jmp r1          ; jump to read
return_point: add r2 r2 r2    ; heavy computation on sensor value
              mov r0 r3
              mov r3 0
              jmp r0          ; return to caller
fails:        fail
transformer_end: data_transformer:
  (SU, ot_xform_enc, ot_xform_enc+2, ot_xform_enc)
  (RW, mmio_b, mmio_e, mmio_a)
data_transformer_end:
\end{capasmsmall}
\caption{Full Cerisier code of the transformer enclave}
\label{fig:code_trusted_sensor_readout_transformer}
\end{figure}


\clearpage

\section{Support for Capability Sealing in the Logical Relation}
\label{app:sealing}
The resource $\sealPred(\otype)(P)$ is a persistent resource keeping track of the predicate $P$ registered for the otype $\otype$.
Technically, it is defined as an instance of \iris{} saved predicates, indexed by otypes.
Only a single predicate can be registered for a given otype:
\begin{lemma}[Seal predicate agree]
  \label{lem:seal-pred-agree}
\[
  \sealPred(\otype)(P_{1}) \ast \sealPred(\otype)(P_{2})
  \wand \later ( P_{1} \equiv P_{2} )
\]
\end{lemma}

A sealed capability $\sealcap$ is safe to share if the capability $\V{sc}$ respects the sealing predicate associated with the otype $o$ (see \Cref{fig:logrel}).
If one already knows the predicate for $\otype$ (\ie~$\sealPred(\otype)(P)$), then it follows from \Cref{lem:seal-pred-agree} that the sealed value $\V{sc}$ respects $P$:
\begin{lemma}[Safe to share sealed agree]
  \label{lem:safe-to-share-sealed-agree}
\[
  \sealPred(\otype)(P) \ast \safeV{\sealcap} \wand \later P(\V{sc})
\]
\end{lemma}

Whether a sealing capability ${[\V{sp},\otype_{\X{b}},\otype_{\X{e}},\otype_{\X{i}}]}$ is safe to share, depends on the permission $\V{sp}$.
If it has permission to seal then we require $\safeseal(\otype_{\X{b}},\otype_{\X{e}})$: the sealing predicates for otypes in $\crange{\otype_b}{\otype_b}$ must hold for any values the adversary might have access to, \ie{} safe-to-share values.
Dually, if $\V{sp}$ includes permission to unseal, then we require $\safeunseal(\otype_b,\otype_e)$: the sealing predicates for otypes in $\crange{\otype_b}{\otype_b}$ must imply safety of sealed values, since the adversary can unseal them.

Finally, if $\V{sp}$ includes both unsealing and sealing permission, then the capability can also be used to $\instr{deinit}$ the enclave associated to the sealing capability (if any).
As such, we additionally require $\safeattest(\otype_{\X{b}},\otype_{\X{e}})$: if the otype \otype{} is associated to an enclave index \V{tidx}, then an invariant must own either the deinitialization authority for \V{tidx} or \V{tidx} is already deinitialized.

\section{Technical Definitions of the Program logic}
\label{sec:appendix_lversion_invariant}
\subsection{State Interpretation for Revocation}
Formally, the invariant maintains two relations.
The correspondence $\V{lr \phyLogReg{vm} phr}$ between
the logical register file $\V{lr}$
and the physical register file $\V{phr}$
for a given view map $\V{vm}$
states that
\begin{enumerate}
  \item for all the registers in the logical register file,
        the projection of the logical word to physical words
        matches with the words in the physical register file;
  \item all logical words in the register file
        always point to the current version of the addresses they can access.
\end{enumerate}
The correspondence $\V{lm \phyLogMem{vm} phm}$ between
the logical memory $\V{lm}$
and the physical memory $\V{phm}$
for a given view map $\V{vm}$
states that
\begin{enumerate}
  \item for all logical addresses in the logical memory,
        the version number is always lower or equal to the current version number,
        and the current version of the same address also exists in the logical memory;
  \item for all addresses in the view map,
        the logical address, for the current version, exists in the logical memory,
        it contains a word that is a root,
        and which projection to physical word matches with the content in the physical memory;
  \item for all reserved addresses in $\reservedAddresses$, the version number is fixed to some version
        $\vinit$.
\end{enumerate}
\[
  \begin{array}{l}
    \V{\versionMono(lm)(vm)} \eqdef \! \! \\
    \quad
    \forall \V{(a,v)} \in \V{dom(lm)}.\;
    v \le \V{vm[a]}
    \land
    (a, \V{vm[a]}) \in \V{dom(lm)}
    \vspace{1em}
    \\
    \V{\reachableCurrent(phm)(lm)(vm)} \eqdef \!\! \\
    \quad
    \left(
      \begin{array}{l}
      \forall a \in \V{dom(vm)}.\; \exists \V{lw}.\;
      \V{lm(a,vm[a])} = \V{lw} \\
      {} \land \V{\iscurrent(vm)(lw})
      \land
      \V{phm[a] = \stripword(lw)}
      \end{array}
      \right)
    \vspace{1em}
    \\
    \V{lr \phyLogReg{vm} phr} \eqdef \!\! \\
    \quad
    \V{\stripreg(lr) = phr} \land
    \left( \forall (\_,\V{lw}) \in \V{lr}.\; \V{\iscurrent(vm)(lw}) \right)
    \\
    \vspace{0em}
    \\
    \V{lm \phyLogMem{vm} phm} \eqdef \!\! \\
    \quad
    \left(
    \begin{array}{l}
      \V{\versionMono(lm)(vm)} \\
      {} \land \V{\reachableCurrent(phm)(lm)(vm)} \\
      {} \land
      \left(
      \forall a \in \reservedAddresses.\; \V{vm}[a] = \vinit
      \right)
    \end{array}
    \right)
\end{array}
\]
\label{fig:lversion_invariant}


\subsection{Enclaves Resources}
\begin{figure}[!ht]
  \centering
  \[
    \begin{array}{lcl}
      \enclaveLiveAuth(\V{live\_table})
      & \eqdef & \ownGhost{\gamma_{\X{Cur}}}{\authfull(\V{table} : \TIndex \rightharpoonup \SC{Ex}(\EId))} \\
      \enclaveLive{tidx}{I}
      & \eqdef & \ownGhost{\gamma_{\X{Cur}}}{\authfrag(\V{tidx} \mapsto I)} \\
      \enclavePrevAuth(\V{prev\_table})
      & \eqdef & \ownGhost{\gamma_{\X{Prev}}}{\authfull(\V{prev\_table} : \TIndex \rightharpoonup \SC{Ag}(\EId)))} \\
      \enclavePrev{tidx}
      & \eqdef & \exists I \ldotp \ownGhost{\gamma_{\X{Prev}}}{\authfrag(\V{tidx} \mapsto I)} \\
      \enclaveHistAuth(\V{all\_table})
      & \eqdef & \ownGhost{\gamma_{\X{All}}}{\authfull(\V{all\_table} : \TIndex \rightharpoonup \SC{Ag}(\EId)))} \\
      \enclaveHist{tidx}{I}
      & \eqdef & \ownGhost{\gamma_{\X{All}}}{\authfrag(\V{tidx} \mapsto I)}
    \end{array}
  \]
  \caption{\label{fig:enclave_resources} Definition of the enclaves resources. }
\end{figure}


\clearpage

\section{Adequacy}
\label{sec:appendix_adequacy}

Our specifications are written and proven in the \cerisier{} program logic.
To draw end-to-end conclusions about executions of our case studies in terms of only the operational semantics, we prove an \emph{adequacy theorem}.

The theorem is designed to derive invariant properties about individual memory locations, particularly an assertion flag stored at address $\aflag$.
Such properties should follow from an invariant $\knowInv{}{ \aflag \amapsto{\vinit} 0}$ but only if we can guarantee that the logical address cannot be revoked using instructions like $\instr{isunique}$ or $\instr{einit}$.
To this end, we single out a set of reserved addresses $\reservedAddresses$ which trusted code cannot revoke and which the adversary must not be given access to.
In return, the adequacy theorem can derive invariant properties about these addresses.
For simplicity, we assume here that $\reservedAddresses = \{ \aflag \}$, \ie~there is a single reserved address in memory that stores a flag keeping track of assertions.

The adequacy theorem states that if the machine starts in a suitable initial
state, and if the assertion flag $\aflag$ is initially $0$, then for any possible step of
execution, the assertion flag will remain $0$. In other words, none of the
asserts in the code will fail.

\begin{theorem}[End-to-end theorem for enclave attestation]
  \label{thm:adequacy}
  \newcommand{\bassert}{\V{b_{a}}}
  \newcommand{\eassert}{\V{e_{a}}}
  \newcommand{\badv}{\V{b_{adv}}}
  \newcommand{\eadv}{\V{e_{adv}}}
  \newcommand{\btcenclave}{\V{b_{tc\_enclave}}}
  \newcommand{\etcenclave}{\V{e_{tc\_enclave}}}

  Suppose disjoint memory fragments
  $\V{prog} : \crange{\V{b}}{\V{e}} \rightarrow \Word$,
  $\V{assert} : \crange{\bassert}{\eassert} \rightarrow \Word$,
  $\V{adv} : \crange{\badv}{\eadv} \rightarrow \Word$.

\begin{enumerate}
  \item the initial state of memory $\V{mem}$ satisfies:
        $
        \V{prog} \uplus \V{assert} \uplus \V{adv} \subseteq \V{mem}
        $
  \item $\crange{\V{b}}{\V{e}}$ contains the verifier's program;
  \item $\crange{\bassert}{\eassert}$ contains the \instr{assert} routine and its flag at address $\aflag$;
  \item the assertion flag is initially set to 0: $\V{mem}(\aflag) = 0$
  \item $\V{prog}$ contains a table linking $\V{assert}$:\\
    $\exists \V{data}, \V{table} \in \crange{b}{e},
            \V{mem}(\V{data}) =\ (\perm{ro}, \V{table}, \V{table} + 1, \V{table}) \land
            \V{mem}(\V{table}) =\ (\perm{e}, \bassert, \eassert, \bassert)$
  \item the adversary region contains only integer and capabilities pointing to its own region:\\
    $\forall a \in \dom(\V{adv}), \V{adv}(a) \in \ZZ \vee \X{inRegion}(\V{adv}(a),\badv,\eadv)$
  \item the initial state of registers $\V{reg}$ satisfies:\\
  $
    \V{reg}(\pc) = (\perm{rwx}, b, e, b), \quad
    \V{reg}(\reg{0}) = (\perm{rwx}, \badv, \eadv, \badv), \quad
    \V{reg}(r) \in \ZZ \; \text{, otherwise}
  $
  \item the initial state of the enclave table $\V{etbl}$ and the enclave counter $\V{ec}$ satisfies:\\
  $
          \forall \V{tidx} \in \dom(\V{etbl}), \V{tidx} < \V{ec} \quad \quad
          2 * \V{ec} \leq \OTypeMax
 $
  \item under the condition that $\aflag$ is the only reserved address ($\reservedAddresses = \{ \aflag \}$),
        the proof in the program logic that the initial configuration is safe given the invariants:
  \[
  \begin{array}{l}
  \exists \customEnclaves \ldotp \forall \V{reg} \ldotp
    \xxsmallh{a}~\knowInv{}{\X{assertInv}(\V{b_a}, \V{e_a}, \vinit, \aflag)},
    \xxsmallh{b}~\knowInv{}{ \aflag \amapsto{\vinit} 0},
    \xxsmallh{c}~\knowInv{}{ \systemInvariant(\customEnclaves) }
  \\[0.2em]
  \displaystyle
  \quad
  \vdash
    \left\{
    \begin{array}{l}
    \xxsmallh{d}~\pc \rmapsto (\perm{rwx}, b, e, b, \vinit) \\
    {} \ast \xxsmallh{e}~\reg{0} \rmapsto (\perm{rwx}, \badv, \eadv, \badv, \vinit) \\[0.3em]
    {} \ast \xxsmallh{f}~\Sep_{\substack{(r, w) \in \V{reg}, \\ r \notin \{\pc, \reg{0}\}}} r \rmapsto z \ast \pure{z \in \ZZ} \\
     {} \ast \xxsmallh{g}~\Sep_{\substack{(a, w) \in \V{prog}, \\ a \notin \{ \V{data},\V{table} \} }} a \amapsto{\vinit} w \\
     {} \ast \xxsmallh{h}~\V{data} \amapsto{\vinit} (\perm{ro}, \V{table}, \V{table} + 1, \V{table}, \vinit)\\
     {} \ast \xxsmallh{i}~\V{table} \amapsto{\vinit} (\perm{e}, \bassert, \eassert, \bassert, \vinit) \\
     {} \ast \xxsmallh{j}~\Sep_{\substack{( a, z) \in \V{adv} }} a \amapsto{\vinit} z \\[0.2em]
   \end{array}
    \right\} \rightsquigarrow \scaleobj{0.7}{\bullet}
  \end{array}
        \]
\end{enumerate}

\vspace{0.5em}

Then, for any $\V{reg'}$, $\V{mem'}$, $\V{etbl'}$ and $\V{ec'}$, if
$(\V{reg}, \V{mem}, \V{etbl}, \V{ec}) \longrightarrow^* (\V{reg'}, \V{mem'}, \V{etbl'}, \V{ec'})$, then
$\V{mem'}(\aflag) = 0$.

\end{theorem}

\begin{proof}
The proof follows from the \iris{} adequacy theorem
and the \cerise~\citep{cerise} adequacy theorem,
but requires non-trivial reasoning,
due to the logical memory added for revocation.
The details of the proof can be found in the \rocq{} development.
\end{proof}

Conditions (1-6) are conditions over the initial memory layout,
assuming that the main (verifier's) program,
the adversary program,
and the assertion routine are properly initialized.
Condition (7) states that initially, $\pc$ contains the verifier's program,
the register $\reg{0}$ contains the capability pointing to the adversary,
and all the other registers contain integers.
Condition (8) states that the initial enclave count is greater than
the highest index in the enclave table,
and there are still some available otypes.
Finally, Condition (9) gives a full specification, in the program logic
of the end-to-end system,
under the condition that the version of the assertion flag cannot be updated.
Each predicate in the program logic reflects
conditions on the initial state of the machine in the operational semantics:
\begin{itemize}
  \item \xxsmallh{a} is the invariant for the closure of the \instr{assert} routine.
        It reflects Condition (3).
  \item \xxsmallh{b} is the invariant stating that the assertion flag is always 0.
        It reflects Condition (4).
  \item \xxsmallh{c} is the system invariant  described in \Cref{sec:ftlr}.
        It essentially reflects Condition (8).
  \item \xxsmallh{d}, \xxsmallh{e} and \xxsmallh{f} are the register points-to predicates.
        It reflects the initial state of the register file described in Condition (7).
  \item \xxsmallh{g} contains the points-to predicates of the program's memory region,
        excluding the data address and linking table.
        $\V{prog}$ is a map from addresses to instructions of the program,
        made concrete when instantiating the theorem.
        \xxsmallh{h} is the points-to predicate of the data address,
        containing a capability pointing to the linking table.
        \xxsmallh{i} contains the memory region of the linking table,
        containing the sentry capability to the \instr{assert} routine.
        All together, they reflect Conditions (2) and (5).
  \item Finally, \xxsmallh{j} is the memory region of the adversary.
        It reflects Condition (6).
\end{itemize}

When using the theorem in the case studies of \Cref{sec:case_studies},
we always assume Conditions (1-8),
which are assumptions about the initial state of the machine;
but Condition (9) is proved.

\end{document}